\documentclass[11pt]{article}

\usepackage{ae}

\usepackage{wide}%

\usepackage{salgorithm}%

\usepackage{graphicx}

\usepackage[thmmarks]{ntheorem}%
\theoremseparator{.}%

\usepackage{titlesec}
\titlelabel{\thetitle. }

\usepackage{picins}

\usepackage{bbm}

\usepackage{hyperref}%
\hypersetup{%
   breaklinks,%
   ocgcolorlinks,
   colorlinks=true,%
   linkcolor=[rgb]{0.45,0.0,0.0},%
   citecolor=[rgb]{0,0,0.45}
}

\usepackage[hypcap=true]{caption}
\usepackage{paralist}%

\usepackage{amsmath}%
\usepackage{color}%
\usepackage{xspace}%
\usepackage{euscript}%
\usepackage{amstext}
\usepackage{amssymb}%

\numberwithin{figure}{section}

\definecolor{blue25}{rgb}{0,0,0.7}
\newcommand{\emphic}[2]{%
     \textcolor{blue25}{%
         \textbf{\emph{#1}}}%
         \index{#2}}
\newcommand{\emphi}[1]{\emphic{#1}{#1}}

\newcommand{\obslab}[1]{\label{observation:#1}}
\newcommand{\obsref}[1]    {Observation~\ref{observation:#1}}
\newcommand{\obsrefpage}[1]{Observation~\ref{observation:#1}%
   $_\text{p\pageref{observation:#1}}$}

\newcommand{\apndlab}[1]{\label{apnd:#1}}
\newcommand{\apndref}[1]{Appendix~\ref{apnd:#1}}

\providecommand{\lemlab}[1]{\label{lemma:#1}}
\providecommand{\lemref}[1]{Lemma~\ref{lemma:#1}}
\newcommand{\lemrefpage}[1]{Lemma~\ref{lemma:#1}%
   $_\text{p\pageref{lemma:#1}}$}

\newcommand{\figlab}[1]{\label{fig:#1}}
\newcommand{\figref}[1]{Figure~\ref{fig:#1}}
\newcommand{\figrefpage}[1]{Figure~\ref{fig:#1}%
   $_\text{p\pageref{fig:#1}}$}

\newcommand{\Term}[1]{\textsf{#1}}
\newcommand{\TermI}[1]{\Term{#1}\index{#1@\Term{#1}}}

\newcommand{\aftermathA}{\par\vspace{-\baselineskip}}

\newcommand{\myqedsymbol}{\rule{2mm}{2mm}}

                  {\unskip\nobreak\hskip 1em plus 1fil\nobreak%
                           \myqedsymbol
                           \parfillskip=0pt%
                           \endtrivlist}
\newcommand{\thmlab}[1]{{\label{theo:#1}}}
\newcommand{\thmref}[1]{Theorem~\ref{theo:#1}}
\newcommand{\thmrefpage}[1]{Theorem~\ref{theo:#1}%
   $_\text{p\pageref{theo:#1}}$}

\newcommand{\eqlab}[1]{\label{equation:#1}}
\newcommand{\Eqref}[1]{Eq.~(\ref{equation:#1})}
\newcommand{\seclab}[1]{{\label{section:#1}}}
\newcommand{\secref}[1]{Section~\ref{section:#1}}
\newcommand{\remlab}[1]{\label{remark:#1}}
\newcommand{\remref}[1]{Remark~\ref{remark:#1}}

\newcommand{\deflab}[1]{\label{defn:#1}}
\newcommand{\defref}[1]{Definition~\ref{defn:#1}}

\newcommand{\ccondlab}[1]{\label{c:cond:#1}}
\newcommand{\ccondref}[1]{(C\ref{c:cond:#1})}
\newcommand{\ccondrefpage}[1]{(C\ref{c:cond:#1})%
   $_\text{p\pageref{c:cond:#1}}$}

\newcommand{\pcondlab}[1]{\label{p:cond:#1}}
\newcommand{\pcondref}[1]{(P\ref{p:cond:#1})}
\newcommand{\pcondrefpage}[1]{(P\ref{p:cond:#1})%
   $_\text{p\pageref{p:cond:#1}}$}

\newcommand{\MakeBig}{\rule[-.2cm]{0cm}{0.4cm}}
\newcommand{\brc}[1]{\left\{ {#1} \right\}}
\newcommand{\sep}[1]{\,\left|\, {#1} \MakeBig\right.}
\newcommand{\pth}[2][\!]{#1\left({#2}\right)}
\newcommand{\epsA}{\tau}%
\newcommand{\nA}{m}%
\newcommand{\DS}{\mathsf{D}}%
\newcommand{\ceiling}[1]{\left \lceil {#1} \right \rceil}
\newcommand{\norm}[1]{\left\lVert {#1} \right \rVert}

\newcommand{\floor}[1]{\left\lfloor {#1} \right\rfloor}
\newcommand{\nfrac}[2]{{#1}/{#2}}
\newcommand{\cardin}[1]{\left\lvert {#1} \right\rvert}
\newcommand{\order}[1]{O\pth{#1}}

\newcommand{\ordereq}[1]{\Theta \left ( {#1} \right )}
\newcommand{\ordergeq}[1]{\Omega \left ( {#1} \right )}
\newcommand{\eps}{{\varepsilon}}%

\newcommand{\WSPD}{\TermI{WSPD}\xspace}
\newcommand{\NN}{\TermI{NN}\xspace}
\newcommand{\ANN}{\TermI{ANN}\xspace}

\newcommand{\AVD}{\TermI{AVD}\xspace}
\newcommand{\MEB}{\TermI{MEB}\xspace}
\newcommand{\etal}{\textit{et~al.}\xspace}
\renewcommand{\Re}{{\rm I\!\hspace{-0.025em} R}}
\newcommand{\diameter}[1]{\mathsf{diam}\pth{ {#1} }}

\newcommand{\PntSet}{\mathsf{P}}
\newcommand{\PntSetB}{\mathsf{Q}}
\newcommand{\PntSetC}{\mathsf{S}}

\newcommand{\FuncSet}{\mathcal{F}}
\newcommand{\FuncSetA}{\mathcal{G}}
\newcommand{\FuncSetB}{\mathcal{H}}

\newcommand{\query}{\mathtt{q}}

\newcommand{\pnt} {\mathsf{p}}
\newcommand{\pntA}{\mathsf{u}}
\newcommand{\pntB}{\mathsf{v}}
\newcommand{\pntC}{\mathsf{s}}
\newcommand{\pntD}{\mathsf{t}}
\newcommand{\distM}[2]{\mathsf{d}\pth{#1,#2}}
\newcommand{\dist}[2]{\norm{{#1}- {#2}}}

\newcommand{\radA}{\mathsf{r}}

\newcommand{\Otilde}{\widetilde{O}}
\newcommand{\remove}[1]{}

\newcommand{\dotP}[2]{\left \langle #1, #2 \right \rangle}

\providecommand{\si}[1]{#1}

\newcommand{\func}{f}
\newcommand{\funcA}{g}

\newcommand{\constA}{c_a}

\newcommand{\constC}{c}
\newcommand{\grconst}{\zeta}
\newcommand{\grfunc}{\lambda}

\newcommand{\subLevel}[2]{{#1}_{\preceq {#2}}}

\newcommand{\st}{~\vert~}

\DeclareMathAlphabet{\mathpzc}{OT1}{pzc}{m}{it}

\newcommand{\DistFromFunc}{\mathbbm{d}}

\newcommand{\sepM}[3][\!]{\DistFromFunc\pth[#1]{#2,#3}}

\newcommand{\sepMA}[3][\!]{\DistFromFunc_{\approx}\pth[#1]{#2,#3}}

\newcommand{\Tree}{T}
\newcommand{\HTree}{\ensuremath{\mathcal{H}}\xspace}

\newcommand{\LHST}[1]{\mathrm{L}_\HTree\pth{#1}}
\newcommand{\LHSTX}[2]{\mathrm{L}_{#1}\pth{#2}}

\newcommand{\Mplus}{\oplus}
\newcommand{\ball}[2]{\mathsf{B}\pth{#1,#2}}

\newcommand{\ballE}[2]{\mathsf{ball}\pth{#1,#2}}

\newcommand{\wt}{w}
\newcommand{\CellSet}{\mathcal{C}}
\newcommand{\CellSetA}{\mathcal{A}}
\newcommand{\cell}{\mathsf{c}}
\newcommand{\MST}{\TermI{MST}\xspace}
\newcommand{\llimit}{\alpha}
\newcommand{\ulimit}{\beta}
\newcommand{\oset}{\alpha}
\newcommand{\setA}{A}
\newcommand{\setB}{B}
\newcommand{\setC}{C}
\newcommand{\obj}{O}
\newcommand{\ObjSet}{\mathcal{O}}
\newcommand{\iradius}{r}

\newcommand{\fatness}{\alpha}
\newcommand{\bdry}[1]{\partial #1}

\newcommand{\objfunc}[1]{F_{#1}}

\newcommand{\polylog}{\text{polylog}~}

\newcommand{\yes}{\TermI{yes}\xspace}
\newcommand{\no}{\TermI{no}\xspace}
\newcommand{\Tleq}[1]{T_{\leq}(#1)}
\newcommand{\Tirq}[1]{T_{r}(#1)}

\newcommand{\refines}{\sqsubseteq}

\newcommand{\CHX}[1]{\mathcal{CH}\pth{#1}}
\renewcommand{\th}{th\xspace}

\newcommand{\num}{\ell}
\newcommand{\numA}{x}

\newcommand{\DSnn}{\mathcal{D}_{near}}
\newcommand{\Grid}{\mathsf{G}\index{grid}}
\newcommand{\Cell}{\Box}

\providecommand{\ds}{\displaystyle}

\newcommand{\mtrA}{\mathcal{X}}
\newcommand{\nnA}{\mathsf{n}_{\query}}
\newcommand{\algonnA}{y}
\newcommand{\approxOrder}[1]{\widetilde{O}\pth{#1}}
\newcommand{\Conv}{C}

\newcommand{\Qtree}{\mathcal{T}}
\newcommand{\Partn}[2]{\left \langle {#1} \right \rangle_{#2}}
\newcommand{\CCS}[2][\!]{\mathsf{C}\pth[#1]{#2}}

\newcommand{\CRinner}{\mathsf{cl}}

\newcommand{\CR}[1]{\CRinner(#1)}

\newcommand{\PF}[3]{\phi\pth{#1,#2,#3}}

\newcommand{\mebr}{\mathrm{z}}
\newcommand{\mebc}{\mathrm{u}}

\newcommand{\SarielThanks}[1]{\thanks{Department of Computer
      Science; 
      University of Illinois; 
      201 N. Goodwin Avenue;
      Urbana, IL, 61801, USA;
      {\tt sariel\atgen{}uiuc.edu}; {\tt
         \url{http://www.uiuc.edu/\string~sariel/}.} #1}}
\newcommand{\NirmanThanks}[1]{\thanks{Department of Computer
      Science; 
      University of Illinois; 
      201 N. Goodwin Avenue;
      Urbana, IL, 61801, USA;
      {\tt \si{nkumar5}\atgen{}illinois.edu}; {\tt
         \url{http://www.cs.uiuc.edu/\string~\si{nkumar5}/}.} #1}}
\newcommand{\atgen}{\symbol{'100}}

\newcommand{\GridCellsX}[2]{\Grid_{\approx #2}\pth{#1}}%

\newcommand{\slGrid}[3]{{ \subLevel{#1}{#3, \approx{#2}}}}

\newcommand{\Union}[1]{\cup #1}%
\newcommand{\DSNearN}[3]{\mathcal{D}_{\mathrm{nr}} \pth{ #1, #2, #3}}%
\newcommand{\obrc}[2][\!\!]{#1\MakeBig \left ( \left. {#2} \MakeBig
       \right ] \right.}%
\newcommand{\pbrc}[2][\!\!]{#1\left[ {#2} \MakeBig \right]}%
\newcommand{\nDS}{L}

\newcommand{\Partition}{\Pi}%
\newcommand{\PartitionA}{\Xi}
\newcommand{\PartitionB}{\Upsilon}

\newcommand{\DFS}{\Algorithm{DFS}\xspace}
\newcommand{\Cluster}{C}
\newcommand{\ClusterA}{D}

\newcommand{\Family}{\EuScript{F}}

\newcommand{\QTree}{\mathcal{Q}}

\newcommand{\constSk}{{\mathsf{c}_{\mathrm{sk}}}}

\newcommand{\corlab}[1]{\label{cor:#1}}
\newcommand{\corref}[1]{Corollary~\ref{cor:#1}}
\newcommand{\correfpage}[1]{Corollary~\ref{cor:#1}%
   $_\text{p\pageref{cor:#1}}$}

\newcommand{\myparagraph}[1]{\noindent \emph{#1}.}

\newcommand{\WeightX}[1]{\omega_{#1}}
\newcommand{\cntr}{\rho}

\renewcommand{\myparagraph}[1]{\paragraph{#1}}%

\newtheorem{theorem}{Theorem}[section] 

\newtheorem{defn}[theorem]{Definition}
\newtheorem{observation}[theorem]{Observation}

\newtheorem{lemma}[theorem]{Lemma}
\newtheorem{corollary}[theorem]{Corollary}

\newtheorem{remark}[theorem]{Remark}%
\theoremstyle{remark}{\theorembodyfont{\rm}
   \newtheorem{example}[theorem]{Example}
}

\theoremheaderfont{\em}%
\theorembodyfont{\upshape}%
\theoremstyle{nonumberplain}%
\theoremseparator{}%
\theoremsymbol{\rule{2mm}{2mm}}%
\newtheorem{proof}{{P}roof:}%

\newcommand{\fmin}{\func_{\min}}
   
\newcommand{\Ell}{\EuScript{E}}

\newcommand{\Decider}{\mathcal{D}}

\newcommand{\RestateTheorem}[2]{%
   \begin{trivlist}
       \item[] \noindent\textbf{Restatement of \thmrefpage{#1}.}
       {#2}
   \end{trivlist}}

\newcommand{\RestateGeneric}[2]{%
   \begin{trivlist}
       \item[] \noindent\textbf{Restatement of #1.}
       {#2}
   \end{trivlist}}

\newcommand{\fn}{\mathcal{F}}
\newcommand{\PartAprxC}{\Psi}%
\newcommand{\PartAprx}[3]{\PartAprxC_{#2} \pth{#1, #3}}%
\newcommand{\separation}{distance$_f$\xspace}
\newcommand{\SearchNaive}{\Algorithm{Search{}}\xspace}
\newcommand{\hRec}{\mathsf{h}}
\newcommand{\funcsX}[1]{\mathrm{F}\pth{#1}}
\newcommand{\IndSet}{{\mathcal I}}

\begin{document}

\title{Approximating Minimization Diagrams and Generalized %
   Proximity Search%
   \footnote{%
      Work on this paper was partially supported by NSF AF award
      CCF-0915984, and NSF AF award CCF-1217462.}%
}%

\author{%
   Sariel Har-Peled\SarielThanks{}%
   \and%
   Nirman Kumar\NirmanThanks{}}%

\date{\today}

\maketitle

\begin{abstract}
    We investigate the classes of functions whose minimization
    diagrams can be approximated efficiently in $\Re^d$.  We present a
    general framework and a data-structure that can be used to
    approximate the minimization diagram of such functions.  The
    resulting data-structure has near linear size and can answer
    queries in logarithmic time. Applications include approximating
    the Voronoi diagram of (additively or multiplicatively) weighted
    points. Our technique also works for more general distance
    functions, such as metrics induced by convex bodies, and the nearest
    furthest-neighbor distance to a set of point sets. 
    Interestingly, our framework works also for
    distance functions that do not comply with the triangle
    inequality. For many of these functions no near-linear size
    approximation was known before.  
\end{abstract}

\section{Introduction}

Given a set of functions $\FuncSet = \brc{\func_i : \Re^d \to \Re
   \sep{ i=1,\ldots, n }}$, their minimization diagram is the function
$\func_{\min}(\query) = \min \limits_{i=1,\ldots, n}
\func_i(\query)$, for any $\query \in \Re^d$. By viewing the graphs of
these functions as manifolds in $\Re^{d+1}$, the graph of the
minimization diagram, also known as the \emphi{lower envelope} of
$\FuncSet$, is the manifold that can be viewed from an observer at
$-\infty$ on the $x_{d+1}$ axis.  Given a set of functions $\FuncSet$
as above, many problems in Computational Geometry can be viewed as
computing the minimization diagram; that is, one preprocesses
$\FuncSet$, and given a query point $\query$, one needs to compute
$\func_{\min}(\query)$ quickly.  This typically requires $n^{O(d)}$
space if one is interested in logarithmic query time. If one is
restricted to using linear space, then the query time deteriorates to
$O\pth{n^{1-O(1/d)}}$ \cite{m-ept-92, c-opt-10}. There is substantial
work on bounding the complexity of the lower envelope in various
cases, how to compute it efficiently, and performing range search
on them; see the book by Sharir and Agarwal \cite{sa-dsstg-95}.

\myparagraph{Nearest neighbor.}  One natural problem that falls into
this framework is the nearest neighbor (\NN) search problem. Here,
given a set $\PntSet$ of $n$ data points in a metric space $\mtrA$, we
need to preprocess $\PntSet$, such that given a query point $\query
\in \mtrA$, one can find (quickly) the point $\nnA \in \PntSet$
closest to $\query$.  Nearest neighbor search is a fundamental task
used in numerous domains including machine learning, clustering,
document retrieval, databases, statistics, and many others.

To see the connection to lower envelopes, consider a set of data
points $\PntSet = \brc{\pnt_1, \dots, \pnt_n}$ in $\Re^d$. Next,
consider the set of functions $\FuncSet = \brc{\func_1, \dots,
   \func_n}$, where $\func_i(\query) = \norm{\query - \pnt_i}$, for
$i=1,\ldots, n$.  The graph of $\func_i$ is the set of points
$\brc{(\query,\func_i(\query)) \st \query \in \Re^d}$ (which is a cone
in $\Re^{d+1}$ with apex at $(\pnt_i,0)$).  Clearly the \NN problem is
to evaluate the minimization diagram of the functions at a query point
$\query$.

More generally, given a set of $n$ functions, one can think of the
minimization diagram defining a ``distance function'', by analogy with
the above. The distance of a query point here is simply the ``height''
of the lower envelope at that point.

\myparagraph{Exact nearest neighbor.} %
The exact nearest neighbor problem has a naive linear time algorithm
without any preprocessing. However, by doing some nontrivial
preprocessing, one can achieve a sub-linear query time. In $\Re^d$,
this is facilitated by answering point location queries using a
Voronoi diagram \cite{bcko-cgaa-08}.  However, this approach is only
suitable for low dimensions, as the complexity of the Voronoi diagram
is $\ordereq{n^{\ceiling{d/2}}}$ in the worst case.  Specifically,
Clarkson \cite{c-racpq-88} showed a data-structure with query time
$O(\log n )$ time, and $\order{n^{\ceiling{d/2} + \delta}}$ space,
where $\delta > 0$ is a prespecified constant (the $O(\cdot)$ notation
here hides constants that are exponential in the dimension). One can
trade-off the space used and the query time \cite{am-rsps-93}. Meiser
\cite{m-plah-93} provided a data-structure with query time $\order{d^5
   \log n }$ (which has polynomial dependency on the dimension), where
the space used is $\order{ n^{ d + \delta}}$. These solutions are
impractical even for data-sets of moderate size if the dimension is
larger than two.

\myparagraph{Approximate nearest neighbor.} %
In typical applications, however, it is usually sufficient to return
an \emphi{approximate nearest neighbor} (\emphi{\ANN{}}). Given an
$\eps > 0$, a $(1 + \eps)$-\ANN, to a query point $\query$, is a point
$\algonnA \in \PntSet$, such that
\begin{equation*}
    \dist{\query}{\algonnA} \leq (1 + \eps)\dist{\query}{\nnA},
\end{equation*}
where $\nnA \in \PntSet$ is the nearest neighbor to $\query$ in
$\PntSet$. Considerable amount of work was done on this problem, see
\cite{c-nnsms-06} and references therein.

In high dimensional Euclidean space, Indyk and Motwani showed that
\ANN can be reduced to a small number of near neighbor queries
\cite{im-anntr-98, him-anntr-12}. Next, using locality sensitive
hashing they provide a data-structure that answers \ANN queries in
time (roughly) $\approxOrder{n^{1/(1+\eps)}}$ and preprocessing time
and space $\approxOrder{n^{1+1/(1+\eps)}}$; here the
$\approxOrder{\cdot}$ hides terms polynomial in $\log n$ and
$1/\eps$. This was improved to $\approxOrder{n^{1/(1+\eps)^2}}$ query
time, and preprocessing time and space
$\approxOrder{n^{1+1/(1+\eps)^2}}$ \cite{ai-nohaa-08}.  These bounds
are near optimal \cite{mnp-lblsh-06}.

In low dimensions (i.e., $\Re^d$ for small $d$), one can use linear
space (independent of $\eps$) and get \ANN query time $O(\log n +
1/\eps^{d-1})$ \cite{amnsw-oaann-98, h-gaa-11}. The trade-off for this
logarithmic query time is of course an exponential dependence on $d$.
Interestingly, for this data-structure, the approximation parameter
$\eps$ is not prespecified during the construction; one needs to
provide it only during the query.  An alternative approach, is to use
Approximate Voronoi Diagrams (\AVD), introduced by Har-Peled
\cite{h-rvdnl-01}, which is a partition of space into regions, of
near-linear total complexity, typically with a representative point
for each region that is an \ANN for any point in the region. In
particular, Har-Peled showed that there is such a decomposition of
size $O\pth{(n /\eps^d)\log^2 n}$, such that \ANN queries can be
answered in $O(\log n)$ time.  Arya and Malamatos \cite{am-lsavd-02}
showed how to build \AVD{}'s of linear complexity (i.e.,
$O(n/\eps^d)$). Their construction uses Well-Separated Pair
Decomposition \cite{ck-dmpsa-95}. Further trade-offs between query
time and space for \AVD{}'s were studied by Arya \etal
\cite{amm-sttan-09}.

\myparagraph{Generalized distance functions: motivation.} %
The algorithms for approximate nearest neighbor, extend to various
metrics in $\Re^d$, for example the well known $\ell_p$ metrics. In
particular, previous constructions of \AVD{}'s extend to $\ell_p$
metrics \cite{h-rvdnl-01, am-lsavd-02} as well. However, these
constructions fail even for a relatively simple and natural extension;
specifically, multiplicative weighted Voronoi diagrams. Here, every
site $\pnt$, in the given point set $\PntSet$, has a weight
$\WeightX{\pnt}$, and the ``distance'' of a query point $\query$ to
$\pnt$ is $\func_\pnt\pth{\query} =
\WeightX{\pnt}\dist{\query}{\pnt}$. The function $ \func_\pnt$ is the
natural distance function induced by $\pnt$.  As with ordinary Voronoi
diagrams, one can define the weighted Voronoi diagram as a partition
of space into disjoint regions, one for each site $\pnt$, such that in
the region for $\pnt$ the function $\func_\pnt$ is the one realizing
the minimum among all the functions induced by the points of
$\PntSet$.  It is known that, even in the plane, multiplicative
Voronoi diagrams can have quadratic complexity, and the minimizing
distance function usually does not comply with the triangle
inequality. Intuitively, such multiplicative Voronoi diagrams can be
used to model facilities where the price of delivery to a client
depends on the facility and the distance. Of course, this is only one
possible distance function, and there are many other such functions
that are of interest (e.g., multiplicative, additive, etc.).

\myparagraph{When fast proximity and small space is not possible.} %
Consider a set of segments in the plane, and we are interested in the
nearest segment to a query point. Given $n$ such segments and $n$ such
query points, this is an extension of Hopcroft's problem, which
requires only to decide if there is any of the given points on any of
the segments. There are lower bounds (in reasonable models) that show
that Hopcroft's problem cannot be solved faster than $\Omega
\pth{n^{4/3}}$ time \cite{e-nlbhp-96}. This implies that no
multiplicative-error approximation for proximity search in this case
is possible, if one insists on near linear preprocessing, and
logarithmic query time.

\myparagraph{When is fast \ANN possible.} %
So, consider a set of geometric objects where each one of them induces
a natural distance function, measuring how far a point in space is
from this object.  Given such a collection of functions, the nearest
neighbor for a query point is simply the function that defines the
lower envelope ``above'' the query point (i.e., the object closest to
the query point under its distance function). Clearly, this approach
allows a generalization of the proximity search problem. In
particular, the above question becomes, for what classes of functions,
can the lower envelope be approximated up to $(1+\eps)$-multiplicative
error, in logarithmic time? Here the preprocessing space used by the
data structure should be near linear.

\subsection{Our results}
We characterize the conditions that are sufficient to approximate
efficiently the minimization diagram of functions.  Using this
framework, one can quickly, approximately evaluate the lower envelope
for large classes of functions that arise naturally from proximity
problems. Our data-structure can be constructed in near linear time,
uses near linear space, and answers proximity queries in logarithmic
time (in constant dimension).  Our framework is quite general and
should be applicable to many distance functions, and in particular we
present the following specific cases where the new data-structure can
be used: \smallskip
\begin{compactenum}[(A)]
    \item \textbf{Multiplicative Voronoi diagrams. }%
    Given a set of points $\PntSet$, where the $i$\th point $\pnt_i$
    has associated weight $\wt_i > 0$, for $i=1,\ldots, n$, consider
    the functions $\func_i(\query) = \wt_i \norm{\query -
       \pnt_i}$. The minimization diagram for this set of functions, 
    corresponds to
    the multiplicative weighted Voronoi diagram of the points.  The
    approach of Arya and Malamatos \cite{am-lsavd-02} to construct
    \AVD's using \WSPD's fails for this problem, as that construction
    relies on the triangle inequality that the regular Euclidean
    distance posseses, which does not hold in this case.
    
    We provide a near linear space \AVD construction for this case.
    We are unaware of any previous results on \AVD for
    multiplicatively weighted Voronoi diagrams.
    
    \item \textbf{Minkowski norms of fat convex bodies. }%
    Given a bounded symmetric convex body $\Conv$ centered at the
    origin, it defines a natural metric; that is, for points $\pntA$
    and $\pntB$ their distance, as induced by $\Conv$, denoted by
    $\norm{\pntA - \pntB}_{\Conv}$, is the minimum $x$ such that $x
    \Conv + \pntA$ contains $\pntB$.  So, given a set of $n$ data
    points $\PntSet = \brc{\pnt_1, \dots, \pnt_n}$ and $n$ centrally
    symmetric and bounded convex bodies $\Conv_1, \dots, \Conv_n$, we
    define $\func_i(\query) = \norm{\pnt_i - \query}_{\Conv_i}$, for
    $i = 1, \ldots, n$.  Since each point induces a distance by a
    different convex body, this collection no longer defines a metric,
    and this makes the problem significantly more challenging.  In
    particular, existing techniques for \AVD and \ANN cannot be
    readily applied. Intuitively, the fatness of the associated convex
    bodies turns out to be sufficient to approximate the associated
    distance function, see \secref{fat:bodies}. The negative example
    for the case of segments presented above, indicates that this
    condition is also necessary.

    \item \textbf{Nearest furthest-neighbor. }%
    Consider a situation where the given input is uncertain;
    specifically, for the $i$\th point we are given a set of points
    $\PntSet_i \subseteq \Re^d$ where it might lie (the reader might
    consider the case where the $i$\th point randomly chooses its
    location out of the points of $\PntSet_i$). There is a growing
    interest in how to handle such inputs, as real world measurements
    are fraught with uncertainty, see \cite{drs-pddd-09, a-mmud-09,
       aesz-nnsuu-12, aahpy-nnsuu-13} and references therein.  In
    particular, in the worst case, the distance of the query point
    $\query$ to the $i$\th point, is the distance from $\query$ to the
    furthest-neighbor of $\query$ in $\PntSet_i$; that is,
    $\fn_i(\query) = \max_{\pnt \in \PntSet_i}
    \dist{\query}{\pnt}$. Thus, in the worst case, the nearest point
    to the query is $\fn(\query) = \min_i \fn_i(\query)$. Using our
    framework we can approximate this function efficiently, using
    space $\Otilde(n)$, and providing logarithmic query time. Note,
    that surprisingly, the space requirement is independent of the
    original input size, and only depends on the number of uncertain
    points.

\end{compactenum}

\myparagraph{Paper organization.} %
In \secref{prelims} we define our framework and prove some basic
properties. Since we are trying to make our framework as inclusive as
possible, its description is somewhat abstract.  In \secref{build}, we
describe the construction of the \AVD and its associated
data-structure.  We describe in \secref{applications} some specific
cases where the new \AVD construction can be used.  We conclude in
\secref{conclusions}.

\section{Preliminaries}
\seclab{prelims}

For the sake of simplicity of exposition, throughout the paper we
assume that all the ``action'' takes place in the unit cube
$[0,1]^d$. Among other things this implies that all the queries are in
this region. This can always be guaranteed by an appropriate scaling
and translation of space. The scaling and translation, along with the
conditions on functions in our framework, implies that outside the
unit cube the approximation to the lower envelope can be obtained in
constant time.

\subsection{Informal description of the technique}

Consider $n$ points in the plane $\pnt_1, \ldots, \pnt_n$, where the
``distance'' from the $i$\th point to a query $\query$, is the minimum
scaling of an ellipse $\Ell_i$ (centered at $\pnt_i$), till it covers
$\query$, and let $\func_i$ denote this distance function. Assume that
these ellipses are fat. Clearly each function $\func_i$ defines a
deformed cone. Given a query point $\query \in \Re^2$, we are
interested in the first function graph being hit by a vertical ray
shoot upward from $(\query,0)$. In particular, let $\fmin(\query) =
\min \limits_{i=1,\ldots, n} \func_i(\query)$ be the minimization
diagram of these functions.

\parpic[r]{\includegraphics{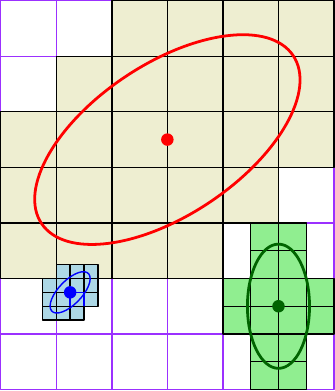}}

As a first step to computing $\fmin(\query)$, consider the decision
version of this problem. Given a value $r$, we are interested in
deciding if $\fmin(\query) \leq r$. That is, we want to decide if
$\query \in \bigcup_i \pth{\pnt_i + r \Ell_i}$. Of course, this is by
itself a computationally expensive task, and as such we satisfy
ourselves with an approximate decision to this procedure. Formally, we
replace every ellipse by a collection of grid cells (of the right
resolution), such that approximately it is enough to decide if the
query point lies inside any of these grid cells -- if it does, we know
that $\fmin(\query) \leq (1+\eps)r$, otherwise $\fmin(\query) > r$.
Of course, as depicted in the right, since the ellipses are of
different sizes, the grid cells generated for each ellipse might
belong to different resolutions, and might be of different
sizes. Nevertheless, one can perform this point-location query among
the marked grid squares quickly using a compressed quadtree.

If we were interested only in the case where $\fmin(\query)$ is
guaranteed to be in 
some interval $[\alpha, \beta]$, then the problem would be easily
solvable. Indeed, build a sequence of the above deciders $\Decider_1,
\ldots, \Decider_m$, where $\Decider_i$ is for the distance
$(1+\eps)^i \alpha$, and $m = \log_{1+\eps} (\beta/\alpha)$. Clearly,
doing a binary search over these deciders with the query point would
resolve the distance query. 

\parpic[r]{\includegraphics{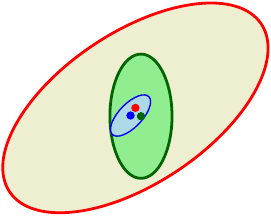}}

\noindent%
\textbf{Sketchable.}  Unfortunately, in general, there is no such
guarantee -- which makes the problem significantly more
challenging. Fortunately, for truly ``large'' distances a collection
of such ellipses looks like a constant number of ellipse (at least in
the approximate case). In the example of the figure above, for large
enough distance, the ellipses looks like a single ellipse, as
demonstrated in the figure on the right. Slight more formally, if
$\bigcup_i \pth{\pnt_i + r \Ell_i}$ is connected, then the set
$\bigcup_i \pth{\pnt_i + R \Ell_i}$ can be $(1+\eps)$-approximated by
a constant number of these ellipses, if $R > \Omega(n r /\eps)$.  A
family of functions having this property is
sketchable. This suggests the problem is easy for very large
distances. 

\paragraph{Critical values to search over.} %
The above suggests that connectivity is the underlying property
that enables us to simplify and replace a large set of ellipses, by a
few ellipses, if we are
looking at them from sufficiently far. This implies that the critical
values when the level-set of the functions changes its connectivity
are the values we should search over during the nearest neighbor
search. Specifically, let $r_i$ be the minimal $r$ when the set
$\bigcup_{k=1}^n \pth{\pnt_k + r_i \Ell_k}$ has $n-i$ connected
components, and let $r_1\leq r_2 \leq \cdots \leq r_n$ be the
resulting sequence. Using the above decision procedure, and a binary
search, we can find the index $j$, such that $r_j \fmin(\query) \leq
r_{j+1}$. Furthermore, the decision procedure for the distance $r_j$,
reports which connected components of $\bigcup_{k=1}^n \pth{\pnt_k +
   r_j \Ell_k}$ contains the query point $\query$. Assume this
connected components is formed by the first $t$ functions; that is,
$\bigcup_{k=1}^t \pth{\pnt_k + r_j \Ell_k}$ is connected and contains
$\query$.  There are two possibilities:
\begin{compactenum}[\quad(A)]
    \item If $\fmin(\query) \in \pbrc{r_{j}, \constA (t/\eps)r_{j}}$,
    then a binary search with the decision procedure would
    approximation $\fmin(\query)$, where $\constA$ is some constant.

    \item If $\fmin(\query)> (t/\eps)r_{j}$ then this whole cluster of
    functions can be sketched and replaced by constant number of
    representative functions, and the nearest-neighbor search can now
    resolve directly by checking for each function in the sketch, what is
    the distance of the query point from it.
\end{compactenum}

\subsubsection{Challenges}

There are several challenges in realizing the above scheme:
\begin{compactenum}[\quad(A)]
    \item We are interested in more general distance functions. To
    this end, we carefully formalize what conditions the underlying
    distance functions induced by each point has to fulfill so that
    our framework applies.

    \item The above scheme requires (roughly) quadratic space to be
    realized. To reduce the space to near linear, we need be more
    aggressive about replacing clusters of points/functions by
    sketches.  To this end, we replace our global scheme by a
    recursive scheme that starts with the ``median'' critical value,
    and fork the search at this value using the decision
    procedure. Now, when continuing the search above this value, we
    replace every cluster (at this resolution) by its sketch.

    \item Computing this ``median'' value directly is too
    expensive. Instead we randomly select a function, we compute the
    connectivity radius of this single distance function with the
    remaining functions. With good probability this value turns out to
    be good.

    \item We need to be very careful to avoid accumulation in the
    error as we replace clusters by sketches.
\end{compactenum}

\subsection{Notations and basic definitions}
\seclab{notation:basic:defns}%

Given $\query \in \Re^d$ and $\PntSet \subseteq \Re^d$ a non-empty
closed set, the \emphi{distance} of $\query$ to $\PntSet$ is
$\distM{\query}{\PntSet} = \min \limits_{x \in \PntSet} \norm{\query -
   x}$.  For a number $\num >0$, the \emphi{grid} of side-length
$\num$, denoted by $\Grid_\num$, is the natural tiling of $\Re^d$,
with cubes of side-length $\num$ (i.e. with a vertex at the origin).
A cube $\Cell$ is \emphi{canonical} if it belongs to $\Grid_\num$,
$\num$ is a power of $2$, and $\Cell \subseteq [0,1]^d$. Informally, a
canonical cube (or cell) is a region that might correspond to a cell
in a quadtree having the unit cube as the root region.

\begin{defn}%
    \deflab{grid:approx}%
    To approximate a set $X \subseteq [0,1]^d$, up to distance $r$,
    consider the set $\GridCellsX{X}{r}$ of all the canonical grid
    cells of $\Grid_\num$ that have a non-empty intersection with $X$,
    where $\num = 2^{\floor{\log_2 (r/\sqrt{d})}}$.  Let $\Union{
       \GridCellsX{X}{r}} = \bigcup_{\Cell \in \GridCellsX{X}{r}}
    \Cell$, denote the union of cubes of $\GridCellsX{X}{r}$.
\end{defn}

Observe that $X \subseteq \Union{ \GridCellsX{X}{r}} \subseteq X
\oplus \ball{0}{r}$, where $\oplus$ denotes the Minkowski sum, and
$\ball{0}{r}$ is the ball of radius $r$ centered at the origin.
\begin{defn}
    For $\num \geq 0$ and a function $\func:\Re^d \rightarrow \Re$,
    the $\num$ \emphi{sublevel set} of $\func$ is the set
    $\subLevel{\func}{\num} = \brc{\pnt \in \Re^d \sep{ \func(\pnt)
          \leq \num}}$.  For a set of functions $\FuncSet$, let
    $\subLevel{\FuncSet}{\num} = \bigcup_{\func \in \FuncSet}
    \subLevel{\func}{\num}$.
\end{defn}

\begin{defn}
    Given a function $\func$ and $\query \in \Re^d$ their
    {\separation} is $\sepM{\query}{\func} = \func(\query)$.  Given
    two functions $\func$ and $\funcA$, their \emphi{\separation{}}
    $\sepM{\func}{\funcA}$ is the minimum $l \geq 0$ such that
    $\subLevel{\func}{l} \cap \subLevel{\funcA}{l} \neq \emptyset$.
    Similarly, for two sets of function, $\FuncSet$ and $\FuncSetA$,
    their \emphi{\separation{}} is
    \begin{align*}
        \sepM{\FuncSet}{\FuncSetA} = \min_{\func \in \FuncSet, \funcA
           \in \FuncSetA} \sepM{\func}{\funcA}.
    \end{align*}
\end{defn}

\begin{example}
    To decipher these somewhat cryptic definitions, the reader might
    want to consider the standard settings of regular Voronoi
    diagrams. Here, we have a set $\PntSet$ of $n$ points. The $i$\th
    point $\pnt_i \in \PntSet$ induces the natural function $\func_i
    \pth{\query} = \dist{\query}{\pnt_i}$.  We have:
    \begin{compactenum}[\qquad(A)]
        \item The graph of $\func_i$ in $\Re^{d+1}$ is a cone
        ``opening upwards'' with an apex at $(\pnt_i, 0)$.
        
        \item The $\num$ sublevel set of $\func_i$ (i.e.,
        $\subLevel{\pth[]{\func_i}}{\num}$) is a ball of radius $\num$
        centered at $\pnt_i$.
        
        \item The \separation of $\query$ from $\func_i$ is the
        Euclidean distance between $\query$ and $\pnt_i$.
        
        \item Consider two subsets of points $X, Y \subseteq \PntSet$
        and let $\FuncSet_X$ and $\FuncSet_Y$ be the corresponding
        sets of functions. The \separation $\num =
        \sepM{\FuncSet_X}{\FuncSet_Y}$ is the minimum radius of balls
        centered at points of $X$ and $Y$, such that there are two
        balls from the two sets that intersect; that is, $\num$ is
        half the minimum distance between a point of $X$ and a point
        of $Y$. In particular, if the union of balls of radius $\num$
        centered at $X$ is connected i.e.
        $\subLevel{\pth[]{\FuncSet_X}}{\num}$ is connected, and
        similarly for $Y$, then $\subLevel{\pth[]{\FuncSet_X \cup
              \FuncSet_Y}}{\num}$ is connected. This is the critical
        value where two connected components of the sublevel set
        merge.
    \end{compactenum}
\end{example}

The \separation function behaves to some extent like a distance
function:
\begin{inparaenum}[(i)]
    \item $\sepM{\func}{\funcA}$ always exists, and
    \item (symmetry) $\sepM{\func}{\funcA} = \sepM{\funcA}{\func}$,
\end{inparaenum}
Also, we have $\subLevel{\func}{\sepM[]{\func}{\funcA}} \neq
\emptyset$.  We extend the above definition to sets of functions.
Note that the triangle inequality does not hold for
$\sepM{\cdot}{\cdot}$.
\begin{observation}%
    \obslab{notbothnear}%
    Suppose that $\func$ and $\funcA$ are two functions such that
    $\sepM{\func}{\funcA} > 0$ and $\query \in \Re^d$. Then,
    $\max\pth{\sepM{\query}{\func},\sepM{\query}{\funcA}} \geq
    \sepM{\func}{\funcA}$.
\end{observation}
\begin{defn}
    \deflab{sets:connected}
    Let $\setB_1, \setB_2, \dots, \setB_m$ be $n$ connected, nonempty 
    sets in $\Re^d$. This collection of sets is \emphi{connected} 
    if $\,\cup_i \setB_i$ is connected.
\end{defn}

\subsubsection{Sketches}
A key idea underlying our approach is that is that any set of
functions of interest should look like a single (or a small number of
functions) from ``far'' enough.  Indeed, given a set of points
$\PntSet \subseteq \Re^d$, they look like a single point (as far as
distance), if the distance from $\CHX{\PntSet}$ is at least
$2\diameter{\PntSet}/\eps$.
\begin{defn}[$\CR{\FuncSet}$]
    \deflab{connect:level}%
    Given a set of functions $\FuncSetA$, if ${\FuncSetA}$ contains a
    single function then the \emphi{connectivity level}
    $\CR{\FuncSetA}$ is $0$; otherwise, it is the minimum $\num \geq
    0$, such that the collection of sets $\subLevel{\func}{\num}$ for
    $\func \in \FuncSetA$ is connected, see \defref{sets:connected}.
\end{defn}

\begin{remark}
  It follows from \defref{connect:level} that
  at level $\num = \CR{\FuncSetA}$, each of the sets $\subLevel{\func}{\num}$
  for $\func \in \FuncSetA$ are nonempty and connected and further their
  union $\subLevel{\FuncSetA}{\num}$ is also connected. 
  This can be relaxed to require that the 
  intersection graph of the sets 
  $\subLevel{\func}{\num}$ for $\func \in \FuncSetA$ is connected 
  (this also implies they are nonempty). Notice that, if at level $\num$, 
  the sublevel sets are connected, then the relaxed definition is 
  equivalent to \defref{connect:level}.
  However, the relaxed definition introduces more technical baggage, and
  for all the interesting applications we have, the sublevel sets
  $\subLevel{\func}{y}$ are connected at all levels $y$ they are nonempty. 
  Therefore, in the interest of brevity, and to keep
  the presentation simple, we mandate that the sublevel sets be connected
  at $\num$. In fact, it would not harm to assume that the sublevel sets
  are connected whenever nonempty.
\end{remark}

\begin{defn}
    \deflab{sketch:def}%
    Given a set of functions $\FuncSetA$ and $\delta \geq 0, y_0 \geq
    0$, a \emphi{$\pth[]{\delta,y_0}$-sketch} for $\FuncSetA$ is a
    (hopefully small) subset $\FuncSetB \subseteq \FuncSetA$, such
    that $ \subLevel{\FuncSetA}{y} \subseteq \subLevel{\FuncSetB}{(1 +
       \delta)y}, $ for all $y \geq y_0$.
\end{defn}
It is easy to see that for any $\FuncSetA, \delta \geq 0, y_0 \geq 0$,
if $\FuncSetB \subseteq \FuncSetA$ is a $(\delta,y_0)$-sketch, then
for any $\delta' \geq \delta, y'_0 \geq y_0, \FuncSetB' \supseteq
\FuncSetB$ it is true that $\FuncSetB'$ is a $(\delta',y'_0)$-sketch
for $\FuncSetA$. Trivially, for any $\delta \geq 0, y_0 \geq 0$, it is
true that $\FuncSetB = \FuncSetA$ is a $(\delta,y_0)$-sketch.

\subsection{Conditions on the functions }
\seclab{conditions}%

We require that the set of functions under consideration satisfy the
following conditions.

\begin{compactenum}[\quad(P1)]
    \item \pcondlab{sublevel:compact}%
    \textbf{Compactness.} For any $y \geq 0$ and $i=1,\ldots,n$, the
    set $\subLevel{\pth[]{\func_i}}{y}$ is compact.
    
    \item \pcondlab{bounded:growth}%
    \textbf{Bounded growth.}  For any $\func \in \FuncSet$, there is a
    function $\grfunc_\func : \Re^+ \to \Re^+$, called the
    \emphi{growth function}, such that for any $y \geq 0$ and $\eps >
    0$, if $\subLevel{\func}{y} \neq \emptyset$, then
    $\grfunc_\func(y) \geq \diameter{\subLevel{\func}{y}}/\grconst$,
    where $\grconst$ is an absolute constant, the \emphi{growth
       constant}, depending only on the family of functions and not on
    $n$ and such that if $\query \in \Re^d$ with
    $\distM{\query}{\subLevel{\func}{y}} \leq \eps\grfunc_\func(y)$,
    then $\func(\query) \leq (1+\eps) y$. This is equivalent to
    $\subLevel{\func}{y} \Mplus \ball{0}{\eps\grfunc_\func(y)}
    \subseteq \subLevel{\func}{(1+\eps)y}$, where $\ball{\pntA}{r}$ is
    the ball of radius $r$ centered at $\pntA$.
    

    \item \pcondlab{sketch:small}%
    \textbf{Existence of a sketch.} %
    Given $\delta > 0$ and a subset $\FuncSetA \subseteq \FuncSet$,
    there is a $\FuncSetB \subseteq \FuncSetA$ with
    $\cardin{\FuncSetB} = \order{\Big. \nfrac{1}{\delta^{\constSk}}}$
    and $y_0 = \order{\CR{\FuncSetA} \pth{\big. \nfrac{
             \cardin{\FuncSetA}}{\delta}}^{\constSk}}$ such that,
    $\FuncSetB$ is an $(\delta,y_0)$-sketch, where $\constSk$ is some
    positive integer constant that depends on the given family of
    functions.
\end{compactenum}

\medskip
\noindent We also require some straightforward properties from the
computation model: \medskip
\begin{compactenum}[\quad(C1)]
    \item \ccondlab{s:computable} $\forall \query \in \Re^d$ and $1
    \leq i \leq n$, the value $\func_i(\query)=\sepM{\query}{\func_i}$
    is computable in $O(1)$ time.

    \item \ccondlab{grid:computable}%
    For any $y \geq 0, r > 0$ and $i$, the set of grid cells
    approximating the sublevel set $\subLevel{\pth[]{\func_i}}{r}$ of
    $\func_i$, that is $\slGrid{\pth{\func_i}}{r}{y} = \GridCellsX{
       \subLevel{\pth{\func_i}}{y}}{r}$ (see \defref{grid:approx}), is
    computable in linear time in its size.
    
    \item \ccondlab{f:s:computable} For any $\func_i, \func_j \in
    \FuncSet, 1 \leq i, j \leq n$ the \separation
    $\sepM{\func_i}{\func_j}$ is computable in $O(1)$ time.
\end{compactenum}

\smallskip

We also assume that the growth function $\grfunc_{\pth[]{\func_i}}(y)$
from Condition \pcondref{bounded:growth} be in fact computable easily
i.e.  in $O(1)$ time.

\begin{remark}%
    \remlab{grid:set:bound}%
    We will use Condition \ccondref{grid:computable} for a given $y$
    and $i$ only for $r$ at least $\Omega\pth{\eps
       \grfunc_{\pth[]{\func_i}}(y)}$ i.e. we will use a grid on the
    sublevel set at a low enough resolution typically $\eps$ times its
    growth function value at that point, which by Condition
    \ccondref{grid:computable} is also $\Omega\pth{\eps
       \diameter{\subLevel{\pth[]{\func_i}}{y}}}$. As such the number
    of grid cells in the grid used is $O(1/\eps^d)$.
\end{remark}

\subsubsection{Properties}

The following are basic properties that the functions under
consideration have. Since these properties are straightforward but
their proof is somewhat tedious, we delegate their proof to
\apndref{p:proof}.

In the following, let $\FuncSet$ be a set of functions that satisfy
the conditions above.  \smallskip
\begin{compactenum}[\quad(L1)]
    \item For any $\func \in \FuncSet$, either $\subLevel{\func}{0} =
    \emptyset$ or $\subLevel{\func}{0}$ consists of a single point.
    (See \lemrefpage{s:l:zero:trivial}.)
    
    \item If $\CR{\FuncSetA} = 0$ for any non-empty subset $\FuncSetA$
    then $\cardin{\FuncSetA} = 1$. (See \defref{connect:level} and
    \obsrefpage{c:l:zero}.)
    
    \item Let $\func \in \FuncSetA$ and $y \geq 0$. For any $\pntA,
    \pntB \in \subLevel{\func}{y}$, we have $\pntA \pntB \subseteq
    \subLevel{\FuncSetA}{(1+\grconst/2)y}$, where $\pntA \pntB$
    denotes the segment joining $\pntA$ to $\pntB$.  (See
    \lemrefpage{segment:contained}.)

    \item Let $\setA_1, \dots, \setA_m \subseteq \Re^d$ be compact
    connected sets, $\pntA \pntB$ be a segment such that $\pntA \pntB \cap
    \setA_i \neq \emptyset$, for $i=1, \ldots, k$ and $\pntA \pntB
    \subseteq \bigcup_{i=1}^k \setA_i$. Then, the sets $\setA_1,
    \ldots, \setA_k$ are connected. (See
    \lemrefpage{segment:compact:cover}.)
    
    \item For any $\FuncSetB \subseteq \FuncSetA \subseteq \FuncSet$,
    $\delta \geq 0$ and $y \geq 0$, such that $\FuncSetB$ is a
    $(\delta,y)$-sketch for $\FuncSetA$, we have that,
    $\CR{\FuncSetB} \leq
    (1+\delta)(1+\grconst/2)\max(y,\CR{\FuncSetA})$.  (See
    \lemrefpage{sketch:conn:level}.)

    \item 
    Let $\FuncSetB \subseteq \FuncSetA \subseteq \FuncSet$, such that
    $\FuncSetB$ is a $(\delta,y_0)$-sketch for $\FuncSetA$ for some
    $\delta \geq 0$ and $y_0 \geq 0$.  Let $\query$ be a point such
    that $\sepM{\query}{\FuncSetA} \geq y_0$. Then we have that
    $\sepM{\query}{\FuncSetB} \leq (1 + \delta)
    \sepM{\query}{\FuncSetA}$.  (See \lemrefpage{sketch:rule}.)
    
\end{compactenum}

\subsubsection{Computing the connectivity level}
\seclab{sketch:time}

We implicitly assume that the above relevant quantities can be
computed efficiently.  For example given some $\delta >0$, and $y_0$
as per the bound in condition \pcondref{sketch:small}, a
$(\delta,y_0)$-sketch can be computed in time
$\order{\cardin{\FuncSetA} / \delta^{\constSk}}$ time. We also assume
that $\CR{\FuncSetA}$ can be computed efficiently without resorting to
the ``brute force'' method. The brute force method computes the
individual \separation of the functions and then computes a \MST on the
graph defined by vertices as the functions and edge lengths as their
\separation. Then $\CR{\FuncSetA}$ is the longest edge of this \MST.

\section{Summary of results}

Our main result is the following, the details of which are delegated
to \secref{build}.

\newcommand{\bodyMain}{%
   Let $\FuncSet$ be a set of $n$ functions in $\Re^d$ that complies
   with our assumptions, see \secref{conditions}, and has sketch
   constant $\constSk \geq d$. Then, one can build a data-structure to
   answer \ANN for this set of functions, with the following
   properties:
    \begin{compactenum}[\qquad(A)]
        \item The query time is $O\pth{ \log n + 1/\eps^{\constSk}}$.

        \item The preprocessing time is $O\pth{ n \eps^{-2\constSk}
           \log^{2\constSk+1} n}$.
        
        \item The space used is $O \pth{ n\eps^{-d-1 -\constSk} \log^2
           n }$.
    \end{compactenum}
}

\begin{theorem}%
    \thmlab{main}%
    \bodyMain
\end{theorem}

One can transform the data-structure into an \AVD, and in the process
improve the query time (the space requirement slightly deteriorates).
See \secref{build} for details.

\newcommand{\bodyMainCor}{%
    Let $\FuncSet$ be a set of $n$ functions in $\Re^d$ that complies
    with our assumptions, see \secref{conditions}, and has sketch
    constant $\constSk \geq d$. Then, one can build a data-structure
    to answer \ANN for this set of functions, with the following
    properties:
    \begin{compactenum}[\qquad(A)]
        \item The improved query time is $O\pth{ \log n }$.
        \item The preprocessing time is $O\pth{ n /\eps^{O(1)}
           \log^{2\constSk+1} n}$.
        \item The space used is $S = O \pth{ n/\eps^{O(1)} \log^2 n }$.
    \end{compactenum}
    
    In particular, we can compute an \AVD of complexity $O(S)$ for the
    given functions. That is, one can compute a space decomposition,
    such that every region has a single function associated with it,
    and for any point in this region, this function is the
    $(1+\eps)$-\ANN among the functions of $\FuncSet$. Here, a region
    is either a cube, or the set difference of two cubes.  }

\begin{corollary}%
    \corlab{main}%
    \bodyMainCor
\end{corollary}

\subsection{Distance functions for which the framework applies}

\subsubsection{Multiplicative distance functions with additive
   offsets}

We are given $n$ points in $\Re^d$, where the point $\pnt_i$ has
weight $\wt_i > 0$, and an \emphi{offset} $\oset_i \geq 0$ associated
with it, for $i=1,\ldots, n$. The \emphi{multiplicative distance with
   offset} induced by the $i$\th point is $\func_i(\query) = \wt_i
\norm{\query-\pnt_i} + \oset_i$.  In \secref{mult:plus:offset} we
prove that these distance functions comply with the conditions of
\secref{conditions}, and in particular we get the following result.

\newcommand{\bodyMultOffestMain}{%
   Consider a set $\PntSet$ of $n$ points in $\Re^d$, where the $i$\th
   point $\pnt_i$ has additive weight $\oset_i \geq 0$ and
   multiplicative weight $\wt_i > 0$. The $i$\th point induces the
   additive/multiplicative distance function $\func_i(\query) = \wt_i
   \norm{\query-\pnt_i} + \oset_i$.  Then one can compute a
   $(1+\eps)$-\AVD for these distance functions, with near linear
   space complexity, and logarithmic query time. See \thmrefpage{main}
   for the exact bounds.%
}%
\begin{theorem}%
    \thmlab{mult:offset:main}%
    \bodyMultOffestMain{}
\end{theorem}%

\subsubsection{Scaling distance}

Somewhat imprecisely, a connected body $\obj$ centered at a point
$\cntr$ is $\fatness$-rounded fat if it is $\fatness$-fat (that
is, there is radius $r$ such that $\ballE{\cntr}{r} \subseteq \obj
\subseteq \ballE{\cntr}{\fatness r}$), and from any point $\pnt$ on
the boundary of $\obj$ the ``cone'' $\CHX{\ballE{\cntr}{r} \cup \pnt}$
is contained inside $\obj$ (i.e., every boundary point sees a large
fraction of the ``center'' of the object).  We also assume that the
boundary of each object $\obj$ has constant complexity.

For such an object, its \emphi{scaling distance} to a point $\query$
from $\obj$ is the minimum $t$, such that $\query \in t \obj$ (where the
scaling is done around its center $\cntr$).  Given $n$
$\fatness$-rounded fat objects, it is natural to ask for the
Voronoi diagram induced by their scaling distance.

These natural distance functions induced by such a set of objects
complies with the framework of \secref{conditions}, see
\secref{fat:bodies} for details. As such, we get the following result.

\newcommand{\bodyAFat}{%
   Consider a set $\ObjSet$ of $\fatness$-rounded fat objects in
   $\Re^d$, for some constant $\fatness$.  Then one
   can compute a $(1+\eps)$-\AVD for the scaling distance functions
   induced by $\ObjSet$, with near linear space complexity, and
   logarithmic query time. See \thmrefpage{main} and \corref{main} for
   the exact bounds.%
}

\begin{theorem}
    \thmlab{a:fat}%
    \bodyAFat{}
\end{theorem}%

\subsubsection{Nearest furthest-neighbor}

For a set of points $\PntSetC \subseteq \Re^d$ and a point $\query$,
the \emphi{furthest-neighbor distance} of $\query$ from $\PntSetC$, is
$\fn_{\PntSetC}(\query) = \max_{\pntC \in \PntSetC}
\dist{\query}{\pntC}$; that is, it is the furthest one might have to
travel from $\query$ to arrive to a point of $\PntSetC$.  For example,
$\PntSetC$ might be the set of locations of facilities, where it is
known that one of them is always open, and one is interested in the
worst case distance a client has to travel to reach an open facility.
The function $\fn_{\PntSetC}(\cdot)$ is known as the
\emphi{furthest-neighbor Voronoi} diagram, and while its worst case
combinatorial complexity is similar to the regular Voronoi diagram, it
can be approximated using constant size representation (in low
dimensions), see \cite{h-caspm-99}.

Given $n$ sets of points $\PntSet_1, \dots, \PntSet_n$ in $\Re^d$, we
are interested in the distance function $\fn(\query) = \min_i
\fn_{i}\pth{\query}$, where $\fn_i(\query) =
\fn_{\PntSet_i}(\query)$. This quantity arises naturally when one tries
to model uncertainty \cite{aahpy-nnsuu-13}; indeed, let $\PntSet_i$ be
the set of possible locations of the $i$\th point (i.e., the location
of the $i$\th point is chosen randomly, somehow, from the set
$\PntSet_i$). Thus, $\fn_i(\query)$ is the worst case distance to the
$i$\th point, and $\fn(\query)$ is the worst-case nearest neighbor
distance to the random point-set generated by picking the $i$\th point
from $\PntSet_i$, for $i=1,\ldots, n$. We refer to $\fn(\cdot)$ as the
\emphi{nearest furthest-neighbor} distance, and we are interested in
approximating it.

We prove in \secref{f:n:neighbor} that the distance functions $\fn_1,
\ldots, \fn_n$ comply with the conditions of the framework, and we get
the following result.

\newcommand{\bodyFNNMain}{%
   Given $n$ point sets $\PntSet_1, \dots, \PntSet_n$ in $\Re^d$ with
   a total of $m$ points, and a parameter $\eps >0$, one can
   preprocess the points into an \AVD, of size $\Otilde(n)$, for the
   nearest furthest-neighbor distance defined by these point sets.
   One can now answer $(1+\eps)$-approximate \NN queries for this
   distance in $O(\log n)$ time. (Note, that the space and query time
   used, depend only on $n$, and not on the input size.) }

\begin{theorem}
    \thmlab{fnn:main}%
    \bodyFNNMain
\end{theorem}

\section{Constructing the \AVD}
\seclab{build}

The input is a set $\FuncSet$ of $n$ functions satisfying the
conditions of \secref{conditions}, and a number $0 < \eps \leq 1$.  We
preprocess $\FuncSet$, such that given a query point $\query$ one can
compute a $\func \in \FuncSet$, where $\sepM{\query}{\FuncSet} \leq
\sepM{\query}{\func} \leq (1+\eps) \sepM{\query}{\FuncSet}$.

\subsection{Building blocks}

\subsubsection{Near neighbor}

Given a set of functions $\FuncSetA$, a real number $\llimit \geq 0$,
and a parameter $\eps > 0$, a \emphi{near-neighbor} data-structure
$\DSnn = \DSNearN{\FuncSetA}{\eps}{\llimit}$ can decide
(approximately) if a point has \separation larger or smaller than
$\llimit$. Formally, for a query point $\query$ a near-neighbor query
answers \yes if $\sepM{\query}{\FuncSetA} \leq \llimit$, and \no if
$\sepM{\query}{\FuncSetA} > (1+\eps) \llimit$. It can return either
answer if $\sepM{\query}{\FuncSetA}\in \!\! \obrc[]{\llimit,
   (1+\eps)\llimit}$. If it returns \yes, then it also returns a
function $\func \in \FuncSetA$ such that $\sepM{\query}{\func} \leq
(1+\eps)\llimit$.  The query time of this data-structure is denoted by
$\Tleq{m}$, where $m = \cardin{\FuncSetA}$.
\begin{lemma}%
    \lemlab{near:neighbor}%
    Given a set of $\nA$ functions $\FuncSetA \subseteq \FuncSet$,
    $\llimit > 0$ and $\eps > 0$. One can construct a data structure
    (which is a compressed quadtree), of size $\order{\nA /\eps^d }$,
    in $\order{\nA \eps^{-d} \log \pth{\nA/\eps}}$ time, such that
    given any query point $\query \in \Re^d$ one can answer a
    $(1+\eps)$-approximate near-neighbor query for the distance
    $\llimit$, in time $\Tleq{\nA} = \order{\log \pth{\nA /\eps}}$.
\end{lemma}

\begin{proof}
    For each $\func \in \FuncSetA$ consider the canonical grid set
    $\GridCellsX{\subLevel{\func}{\llimit}}{r_f}$, where $r_f = \eps
    \grfunc_\func(\llimit) \geq \eps
    \diameter{\subLevel{\func}{\llimit}}/\grconst$, where
    $\grfunc_\func(\cdot), \grconst$ are the growth function and the
    growth constant, see \pcondref{bounded:growth}.  The sublevel set
    of interest is $\subLevel{\FuncSetA}{\llimit}$ and its
    approximation is $\CellSet = \bigcup_{\func \in \FuncSetA}
    \GridCellsX{\subLevel{\func}{\llimit}}{r_f}$, as the bounded
    growth condition \pcondref{bounded:growth} implies that
    $\subLevel{\func}{\llimit} \subseteq
    \GridCellsX{\subLevel{\func}{\llimit}}{r_f} \subseteq
    \subLevel{\func}{(1+\eps)\llimit}$.  The set of canonical cubes
    $\CellSet$ can be stored in a compressed quadtree $\Qtree$, and
    given a query point we can decide if a point is covered by some
    cube of $\CellSet$ by performing a point location query in
    $\Qtree$.
    
    By \remref{grid:set:bound},
    $\cardin{\GridCellsX{\subLevel{\func}{\llimit}}{r_f}} =
    O\pth{\eps^{-d}}$.  As such, the total number of canonical cubes
    in $\CellSet$ is $O\pth{\nA /\eps^d}$, and the compressed
    quadtree for storing them can be computed in $O\pth{\nA \eps^{-d}
       \log( \nA /\eps)}$ time \cite{h-gaa-11}.
    
    We mark a cell of the resulting quadtree by the function whose
    sublevel set it arose from (ties can be resolved
    arbitrarily). During query, if $\query$ is found in one of the
    cells we return \yes and the function associated with the cell,
    otherwise we return \no.
    
    If we have that $\sepM{\query}{\FuncSetA} \leq \llimit$, then the
    query point $\query$ will be found in one of the marked cells,
    since they cover $\subLevel{\FuncSetA}{\llimit}$. As such, the
    query will return \yes. Moreover, if the query does return a \yes,
    then it belongs to a cube of $\CellSet$ that is completely covered
    by $\subLevel{\FuncSetA}{(1+\eps)\llimit}$, as desired.
\end{proof}

\subsubsection{Interval data structure}

Given a set of functions $\FuncSetA$, real numbers $0 \leq \llimit
\leq \ulimit$, and $\eps > 0$, the \emphi{interval data structure}
returns for a query point $\query$, one of the following:
\begin{compactenum}[\qquad(A)]
    \item If $\sepM{\query}{\FuncSetA}\in \pbrc[]{\llimit, \ulimit}$,
    then it returns a function $\funcA \in \FuncSetA$ such that
    $\sepM{\query}{\funcA} \leq (1+\eps)\sepM{\query}{\FuncSetA} $. It
    might also return such a function for values outside this
    interval.
    
    \item ``$\sepM{\query}{\FuncSetA} < \llimit$''. In this case it
    returns a function $\funcA \in \FuncSetA$ such that
    $\sepM{\query}{\funcA} < \llimit$.
    \item ``$\sepM{\query}{\FuncSetA} > \ulimit$''.
\end{compactenum}
The time to perform an interval query is denoted by
$\Tirq{m,\llimit,\ulimit}$.

\begin{lemma}%
    \lemlab{interval}%
    Given a set of $\nA$ functions $\FuncSetA$, an interval $[\llimit,
    \ulimit]$ and an approximation parameter $\epsA > 0$, one can
    construct an interval data structure of size $\order{\nA
       \epsA^{-d-1} \log(4\ulimit/\llimit) }$, in time $\order{\nA
       {\epsA^{-d-1}} \log(4\ulimit/\llimit)\log\pth{\nA/\epsA}}$,
    such that given a query point $\query$ one can answer
    $(1+\epsA)$-approximate nearest neighbor query for the distances
    in the interval $[\llimit,\ulimit]$, in time $\ds \Tirq{\nA,
       \llimit,\ulimit,f} = \order{\log
       {\frac{\nA\log(4\ulimit/\llimit)}{\epsA}} }$.
\end{lemma}

\begin{proof}
    Using \lemref{near:neighbor}, build a $(1+\epsA/4)$-near neighbor
    data-structure $\DS_i$ for $\FuncSetA$, for distance $r_i =
    (\alpha/2) (1+\epsA/4)^i$, for $i=0, \ldots,
    \nDS=\ceiling{\log_{1+\epsA/4} (4\ulimit/\llimit)} = O\pth{
       \epsA^{-1} \log (4\ulimit/\llimit)}$. Clearly, an interval
    query can be answered in three stages:
    \begin{compactenum}[\qquad(A)]
        \item Perform a point-location query in $\DS_0$. If the answer
        is \yes then $\sepM{\query}{\FuncSetA} < \llimit$. We can also
        return a function $\funcA \in \FuncSetA$ with
        $\sepM{\query}{\funcA} < \llimit$.
        
        \item Similarly, perform a point-location query in
        $\DS_{\nDS}$.  If the answer is \no then
        $\sepM{\query}{\FuncSetA} > \ulimit$ and we are done.
        
        \item It must be that $\sepM{\query}{\FuncSetA} \in
        \pbrc[]{r_i, r_{i+1}}$ for some $i$. Find this $i$ by
        performing a binary search on the data-structures $\DS_0,
        \ldots, \DS_{\nDS}$, for the first $i$ such that $\DS_i$
        returns \no, but $\DS_{i+1}$ returns \yes. Clearly,
        $\DS_{i+1}$ provides us with the desired $(1+\epsA/4)^2$-\ANN
        to the query point.
    \end{compactenum}
    
    \medskip
    
    To get the improved query time, observe that we can overlay these
    compressed quadtrees $\DS_0, \ldots, \DS_{\nDS}$ into a single
    quadtree. For every leaf (or compressed node) of this quadtree
    we compute the original node with the lowest value covering this
    node. Clearly, finding the desired distance can now be resolved by
    a single point-location query in this overlay of quadtrees. The
    total size of these quadtrees is $S = O\pth{ \nDS \pth{\nA
          /\epsA^d}}$, and the total time to compute these quadtrees
    is $T_1 = O\pth{ \nDS \pth{\nA /\epsA^d} \log \pth{ \nA /\epsA}}$,
    and the time to compute their overlay is $O \pth{S \log
       \nDS}$. The time to perform a point-location query in the
    overlayed quadtree is $O( \log S)$.
\end{proof}

\lemref{interval} readily implies that if somehow a priori we know the
nearest neighbor \separation lies in an interval of values of
polynomial spread, then we would get the desired data-structure by
just using \lemref{interval}. To overcome this unbounded spread
problem, we would first argue that, under our assumptions, there are
only linear number of intervals where interesting things happen to the
\separation function.

\subsubsection{Connected components of the sublevel sets}

Given a finite set $X$ and a partition of it into disjoint sets $X =
X_1 \cup \dots \cup X_k$, let this partition be denoted by
$\Partn{X_1,\dots,X_k}{X}$. For $1 \leq i \leq k$, each $X_i$ is a
\emphi{part} of the partition.

\begin{defn}
    For two partitions $P_A = \Partn{A_1,\dots,A_k}{X}$ and $P_B
    = \Partn{B_1, \dots, B_l}{X}$ of the same set $X$, $P_B$ is a
    \emphi{refinement} of $P_A$, denoted by $P_B \refines P_A$, if for
    any $B_i$ there exists a set $A_{j_i}$, such that $B_i \subseteq
    A_{j_i}$.  In the other direction, $P_A$ is a \emphi{coarsening}
    of $P_B$.
\end{defn}

\begin{observation}
    Given partitions $\Partition, \PartitionA$ of a finite set $X$, if
    $\Partition \refines \PartitionA$ then $\cardin{\PartitionA} \leq
    \cardin{\Partition}$.
\end{observation}

\begin{defn}
    Given partitions $\Partition = \Partn{X_1,\dots,X_k}{X}
    \refines \PartitionA = \Partn{X'_1,\dots, X'_{k'}}{X}$, let
    $\PF{\Partition}{\PartitionA}{i}$ be the function that return the
    set of indices of sets in $\Partition$ whose union is $X_i'
    \in \PartitionA$.
\end{defn}

\begin{observation}
    Given partitions $\Partition \refines \PartitionA$ of a set $X$
    with $n$ elements. The partition function
    $\PF{\Partition}{\PartitionA}{\cdot}$ can be computed in $O(n)$
    time. For any $1 \leq i \leq \cardin{\PartitionA}$, the set
    $\PF{\Partition}{\PartitionA}{i}$ can be returned in
    $O\pth{\big.\cardin{\PF{\Partition}{\PartitionA}{i}}}$ time, and
    its size can be returned in $O(1)$ time.
\end{observation}

\begin{defn}
    For $\FuncSetA \subseteq \FuncSet$ and $\num > 0$, consider the
    intersection graph of the sets $\subLevel{\func}{\num}$, for all
    $\func \in \FuncSetA$.  Each connected component is a
    \emphi{cluster} of $\FuncSetA$ at level $\num$. And the partition
    of $\FuncSetA$ by these clusters, denoted by
    $\CCS{\FuncSetA,\num}$, is the \emphi{$\num$-clustering} of
    $\FuncSetA$.
\end{defn}

The values $\num$ at which the $\num$-clustering of $\FuncSet$ changes
are, intuitively, the critical values when the sublevel set of
$\FuncSet$ changes and which influence the \AVD.  These values are
critical in trying to decompose the nearest neighbor search on
$\FuncSet$ into a search on smaller sets.

\begin{observation}
    If $0 \leq a \leq b$ then $\CCS{\FuncSetA,a} \refines
    \CCS{\FuncSetA,b}$.
\end{observation}

The following lemma testifies that we can approximate the
$\num$-clustering quickly, for any number $\num$.
\begin{lemma}%
    \lemlab{connectivity}%
    Given $\FuncSetA \subseteq \FuncSet$, $\num \geq 0$, and $\eps >
    0$, one can compute, in $\order{\frac{m}{\eps^d} \log (m/\eps)}$
    time, a partition $\PartAprxC = \PartAprx{\FuncSetA}{\eps}{\num}$,
    such that $\CCS{\FuncSetA, \num} \refines \PartAprxC \refines
    \CCS{\FuncSetA,(1+\eps)\num}$, where $m = \cardin{\FuncSetA}$.
\end{lemma}

\begin{proof}
    For each $\func \in \FuncSetA$, tile the sublevel sets
    $\subLevel{\pth[]{\func}}{\num}$ by canonical cubes of small
    enough diameter, such that bounded growth condition
    \pcondref{bounded:growth} assures all cubes are inside
    $\subLevel{\pth[]{\func}}{(1+\eps)\num}$. To this end, for $\func
    \in \FuncSetA$, set $r_\func = \eps\grfunc_\func(\num) \geq \eps
    \diameter{\subLevel{\func}{\num}}/\grconst$, and compute the set
    $\CellSet_\func = \bigcup_{\func \in \FuncSetA} \pth{{
          \GridCellsX{\subLevel{ \pth[]{ \func}}{\num} }{r_\func} }}$,
    see \defref{grid:approx}. It is easy to verify that we have that
    \begin{equation}
        \eqlab{sketcheqn}
        \subLevel{\pth[]{\FuncSetA}}{\num}%
        \subseteq%
        \Union{\CellSet_\func}%
        \subseteq%
        \subLevel{\pth[]{\FuncSetA}}{(1+\eps)\num}.
    \end{equation}
    By assumption, we have that $\cardin{\CellSet_\func} =
    O(1/\eps^d)$, and the total number of canonical cubes in all the
    sets $\CellSet_\func$ for $\func \in \FuncSetA$ is $O(m
    /\eps^d)$. We throw all these canonical cubes into a compressed
    quadtree, this takes $O\pth{ (m/\eps^d) \log (m/\eps)}$
    time. Here, every node of the compressed quadtree is marked if it
    belongs to some of these sets, and if so, to which of the
    sets. Two sets $\Union{\CellSet_\func}$ and
    $\Union{\CellSet_\funcA}$ intersect, if and only if there are two
    canonical cubes, in these two sets, such that they overlap; that
    is, one of them is a sub-cube of the other. Initialize a
    union-find data-structure, and traverse the compressed quadtree
    using \DFS, keeping track of the current connected component, and
    performing a union operation whenever encountering a marked node
    (i.e., all the canonical nodes associated with it, are unionized
    into the current connected component). Finally, we perform a union
    operation for all the cells in $\CellSet_\func$, for all $\func
    \in \FuncSetA$. Clearly, this results in the desired connected
    components of the intersection graph of $\Union{\CellSet_\func}$
    (note, that we consider two sets as intersecting only if their
    interiors intersect). Translating each such connected set of
    canonical cubes back to the functions that gave rise to them,
    results in the desired partition.
\end{proof}

\begin{remark}
    The partition $\PartAprxC$ computed by \lemref{connectivity} is
    monotone, that is, for $\num \leq \num'$ and $\eps \leq \eps'$, we
    have $\PartAprx{\FuncSetA}{\eps}{\num} \refines
    \PartAprx{\FuncSetA}{\eps'}{\num'}$. Moreover, for each cluster
    $\Cluster \in \PartAprx{\FuncSetA}{\eps}{\num}$, we have that
    $\CR{\Cluster} \leq (1+\eps)\num$.
\end{remark}

\subsubsection{Computing a splitting distance}

\begin{defn}%
    \deflab{splitting}%
    Given a partition $\PartAprxC = \PartAprx{\FuncSetA}{\eps}{\num}$
    of $\FuncSetA$, with $m = \cardin{\PartAprxC}$ clusters, a
    distance $\numA$ is a \emphi{splitting distance} if $m/4 \leq
    \cardin{\PartAprx{\FuncSetA}{1}{\numA/4}}$ and
    $\cardin{\PartAprx{\FuncSetA}{1}{\numA}} \leq (7/8)m$.
\end{defn}

\begin{lemma}%
    \lemlab{splitting}%
    Given a partition $\PartAprxC = \PartAprx{\FuncSetA}{\eps}{\num}$
    of $\FuncSetA$, one can compute a splitting distance for it, in
    expected $O( n( \log n + t) )$ time, where $n =\cardin{\FuncSetA}$
    and $t$ is the maximum cluster size in $\PartAprxC$.
\end{lemma}
\begin{proof}
    For each cluster $\Cluster \in \PartAprxC$, let $r_\Cluster$ be
    its \separation from all the functions in $\FuncSetA
    \setminus \Cluster$; that is $r_\Cluster = \min_{\func \in
       \Cluster} \min_{\funcA \in \FuncSetA \setminus \Cluster}
    \sepM{\func}{\funcA}$. Note that $r_\Cluster \geq \num$. Now, let
    $r_1 \leq r_2 \leq \cdots \leq r_m$ be these \separation distances
    for the $m$ clusters of $\PartAprxC$. We randomly pick a cluster
    $\Cluster \in \PartAprxC$ and compute $\num' = r_\Cluster$ for it
    by brute force -- computing the \separation of each function of
    $\Cluster$ with the functions of $\FuncSetA \setminus \Cluster$.
    
    Let $i$ be the rank of $\num' = r_\Cluster$ among $r_1,\ldots,
    r_m$. With probability $1/2$, we have that $m/4 \leq i \leq
    (3/4)m$. If so we have that:
    \begin{compactenum}[\qquad(A)]
        \item All the clusters that correspond to $r_{i}, \ldots,
        r_{m}$ are singletons in the partition
        $\PartAprx{\FuncSetA}{1}{\num/4}$, as the \separation
        of each one of these clusters is larger than $\num'$. We
        conclude that $\cardin{\PartAprx{\FuncSetA}{1}{\num'/4}} \geq
        m/4$.
        
        \item All the clusters of $\PartAprxC$ that correspond to
        $r_{1}, \ldots, r_{i}$ are contained inside a larger cluster
        of $\PartAprx{\FuncSetA}{1}{\num'}$ (i.e., they were merged
        with some other cluster). But then, the number of clusters in
        $\PartAprx{\FuncSetA}{1}{\num'}$ is at most $(7/8)m$. Indeed,
        put an edge between such a cluster, to the cluster realizing
        the smallest \separation with it. This graph has at least $e
        \geq m/4$ edges, and it is easy to see that each component of
        size at least $2$ in the underlying undirected graph has the
        same number of edges as vertices. As such the number of
        singleton components is at most $m - e$ while the number of
        components of size at least $2$ is at most $e/2$. It follows
        that the total number of components is at most $m - e/2 \leq
        7m/8$. Since each such component corresponds to a cluster in
        $\PartAprx{\FuncSetA}{1}{\num'}$ the claim is proved.
    \end{compactenum}
    
    Now, compute $\PartAprx{\FuncSetA}{1}{\num'}$ and
    $\PartAprx{\FuncSetA}{1}{\num'/4}$ using \lemref{connectivity}.
    With probability at least half they have the desired sizes, and we
    are done. Otherwise, we repeat the process. In each iteration we
    spend $O( n( \log n + t) )$ time, and the probability for success
    is half. As such, in expectation the number of rounds needed is
    constant.
\end{proof}


\subsection{The search procedure}

\subsubsection{An initial ``naive'' implementation}

\begin{figure}[t]
    \begin{center}
        \begin{algorithmEnv}
            \SearchNaive{}( $\FuncSetA$, $\PartitionB$, $\query$ )\+\\
            \texttt{// $\FuncSetA$: set of functions}\\
            \texttt{// $\PartitionB = \PartAprx{\FuncSetA}{1}{\num}$
               for some
               value $\num$}\\
            \texttt{// Invariant: $\sepM{\query}{\FuncSetA} > \num N$}\\
            \If $\cardin{\PartitionB} = 1$
            \Then\+\\
            \Return $\sepM{\query}{\FuncSetA} = \min_{\func \in
               \FuncSetA} \sepM{\query}{\func}$%
            \qquad \qquad    (*)\-\\
            $\numA \leftarrow $ compute a splitting distance of
            $\PartitionB$, see
            \lemref{splitting}\\
            \\
            \texttt{// Perform an interval approximate nearest }\\
            \texttt{// \qquad neighbor query on the interval
               $[\numA/8N, \numA 8N ]$
            }\\
            \texttt{// \qquad for the set $\FuncSetA$, see
               \lemref{interval}.%
            }\\
            \If $\sepM{\query}{\FuncSetA} \in \pbrc{\numA/8N, \numA
               8N^2 }$ or
            $(1+\tfrac{\eps}{4})$-\ANN found \Then \+\\
            \Return nearest function found by the
            \\
            \qquad\qquad
            $(1+\eps/4)$-approximate interval query.\-\\
            \If $\sepM{\query}{\FuncSetA} < \numA/8N$ \Then \+\\
            $\func \leftarrow $ $2$-approximate near neighbor query on
            $\FuncSetA$ \\
            \qquad\qquad and distance
            $\numA/8$, see \lemref{near:neighbor}.\\
            Find cluster $\Cluster
            \in \PartAprx{\FuncSetA}{1}{\numA/4}$,
            such that $\func \in \Cluster$,\\
            \qquad\qquad\qquad
            see     \lemref{connectivity}. \\
            \Return \SearchNaive{}( $\Cluster$,
            $\PartitionB[\Cluster]$, $\query$ )
            \-\\
            \If $\sepM{\query}{\FuncSetA} > \numA 8N^2 $ \Then \+\\
            \Return \SearchNaive{}( $\FuncSetA$,
            $\PartAprx{\FuncSetA}{1}{\numA N}$, $\query$ ) %
            \qquad \qquad (**)
        \end{algorithmEnv}
    \end{center}
    \vspace{-0.6cm}
    \captionof{figure}{%
       Search algorithm: We are given a query point $\query$, and an
       approximation parameter $\eps> 0$. The quantity $N$ is a
       parameter to be specified shortly. Initially, we call this
       procedure on the set of functions $\FuncSet$ with $\PartitionB$
       being the partition of $\FuncSet$ into singletons (i.e.,
       $\num=0$).  Here, $\PartitionB[\Cluster]$ denotes the partition
       of $\Cluster$ induced by the partition $\PartitionB$.  }
    \figlab{search:naive}
\end{figure}

The search procedure is presented in \figref{search:naive}.

\begin{lemma}%
    \lemlab{search:naive}%
    \SearchNaive{}$(\, \FuncSetA, \PartitionB, \query \, )$ returns a
    function $\func \in \FuncSetA$, such that $\sepM{\query}{\func}
    \leq (1+\eps)\sepM{\query}{\FuncSetA}$. The depth of the recursion
    of \SearchNaive is $\hRec = O( \log n)$, where $n =
    \cardin{\FuncSetA}$.
\end{lemma}

\begin{proof}
    The proof is by induction on the size of
    $\PartitionB$. If $\cardin{\PartitionB} =1$, then the function
    realizing $\sepM{\query}{\FuncSetA}$ is returned, and the claim is
    true.
    
    Let $\numA$ be the computed splitting distance of
    $\PartitionB$. Next, the procedure perform an
    $(1+\eps/4)$-approximate interval nearest-neighbor query for
    $\query$ on the range $[\numA/8N, \numA 8N ]$. If this computed
    the approximate nearest neighbor then we are done.
    
    Otherwise, it must be that either $\sepM{\query}{\FuncSetA} <
    \numA/8N$ or $\sepM{\query}{\FuncSetA} > 8N \numA$, and
    significantly, we know which of the two options it is:
    \begin{compactenum}[\quad(A)]
        \item If $\sepM{\query}{\FuncSetA} < \numA/8N$ then doing an
        approximate near-neighbor query on $\FuncSetA$ and distance
        $\numA/8$, returns a function $\func \in \FuncSetA$ such that
        $\sepM{\query}{\func} \leq \numA/4$. Clearly, the nearest
        neighbor to $\query$ must be in the cluster containing $\func$
        in the partition $\PartAprx{\FuncSetA}{1}{\numA/4}$, and
        \SearchNaive recurses on this cluster. Now, by induction, the
        returned \ANN is correct.
        
        Since $\numA$ is a splitting distance of $\PartitionB$, see
        \defref{splitting}, we have $\cardin{\PartitionB}/4 \leq
        \cardin{\PartAprx{\FuncSetA}{1}{\numA/4}}$ and $\PartitionB
        \refines \PartAprx{\FuncSetA}{1}{\numA/4}$.  As such, since
        $\Cluster$ is one of the clusters of
        $\PartAprx{\FuncSetA}{1}{\numA/4}$, the induced partition of
        $\Cluster$ by $\PartitionB$ (i.e., $\PartitionB[\Cluster]$),
        can have at most $(1-1/4)\cardin{\PartitionB}$ clusters.
        
        \item Otherwise, we have $\sepM{\query}{\FuncSetA} > \numA
        \cdot 8N$.  Since $\numA$ is a splitting distance, we have
        that $\cardin{\PartAprx{\FuncSetA}{1}{\numA}} \leq
        (7/8)\cardin{\PartitionB}$, see \defref{splitting}.  We
        recurse on $\FuncSetA$, and a partition that has fewer
        clusters, and by induction, the returned answer is correct.
    \end{compactenum}
    \medskip%
    In each step of the recursion, the partition shrunk by at least a
    fraction of $7/8$. As such, after a logarithmic number of
    recursive calls, the procedure is done.
\end{proof}

\subsubsection{But where is the beef? %
   Modifying \SearchNaive to provide fast query time}

The reader might wonder how we are going to get an efficient search
algorithm out of \SearchNaive, as the case that $\PartitionB$ is a
single cluster, still requires us to perform a scan on all the
functions in this cluster and compute their \separation from the query
point $\query$.  Note however, we have the invariant that the distance
of interest is polynomially larger than the connectivity level of each
of the clusters of $\PartitionB$.  In particular, precomputing for all
the sets of functions such that (*) might be called on, their
$\eps/8$-sketches, and answering the query by computing the distance
on the sketches, reduces the query time to $O( 1/\eps^{\constSk} +
\log^2 n)$ (assuming that we precomputed all the data-structures used
by the query process). Indeed, an interval query takes $O( \log n)$
time, and there $O( \log n)$ such queries. The final query on the
sketch takes time proportional to the sketch size which is $O\pth{
   1/\eps^{\constSk} }$.

As such, the major challenge is not making the query process fast, but
rather building the search structure quickly, and arguing that it
requires little space.

\subsubsection{Sketching a sketch}

To improve the efficiency of the preprocessing for \SearchNaive, we are
going to use sketches more aggressively.  Specifically, for each of
the clusters of $\PartitionB$, we can compute their $\delta$-sketches,
for $\delta = \eps/(8\hRec) = O( \eps/\log n)$, see
\lemref{search:naive}. From this point on, when we manipulate this
cluster, we do it on its sketch. To make this work set
$N=n^{4\constSk}$, see \pcondrefpage{sketch:small} and
\lemref{sketch:conn:level}.

The only place in the algorithm where we need to compute the sketches,
is in (**) in \figref{search:naive}. Specifically, we compute
$\PartAprx{\FuncSetA}{1}{\numA N}$, and for each new cluster $\Cluster
\in \PartAprx{\FuncSetA}{1}{\numA N}$, we combine all the sketches of
the clusters $\ClusterA \in \PartitionB$ such that $\ClusterA
\subseteq \Cluster$ into a single set of functions. We then compute a
$\delta$-sketch for this set, and this sketch is this cluster from
this point on. In particular, the recursive calls to \SearchNaive{}
would send the sketches of the clusters, and not the clusters
themselves. Conceptually, the recursive call would also pass the
minimum distance where the sketches are active -- it is easy to verify
that we use these sketches only at distances that are far away and are
thus allowable (i.e., the sketches represent the functions they
correspond to, well in these distances).

Importantly, whenever we compute such a new set, we do so for a
distance that is bigger by a polynomial factor (i.e., $N$) than the
values used to create the sketches of the clusters being merged.
Indeed, observe that $\numA > \num$ and as such $\numA N$ is $N$ times
bigger than $\num$ (an upper bound on the value used to compute the
input sketches).

As such, all these sketches are valid, and can be used at this
distance (or any larger distance). Of course, the quality of the
sketch deteriorates. In particular, since the depth of recursion is
$\hRec$, the worst quality of any of the sketches created in this
process is at most $(1+\delta)^\hRec \leq 1+\eps/4$.

Significantly, before using such a sketch, we would shrink it by
computing a $\eps/8$-sketch of it. This would reduce the sketch size
to $O\pth{1/\eps^{\constSk}}$. Note, however, that this still does not
help us as far as recursion - we must pass the larger
$\delta$-sketches in the recursive call of (**).

This completes the description of the search procedure. It is still
unclear how to precompute all the data-structures required during the
search. To do that, we need to better understand what the search
process does.

\subsection{The connectivity tree, and the preprocessing}

Given a set of functions $\FuncSet$, create a tree tracking the
connected components of the \MST of the functions. Formally, initially
we start with $n$ singletons (which are the leafs of the tree) that
are labeled with the value zero, and we store them in a set $\Family$
of active nodes. Now, we compute for each pair of sets of functions
$X, Y \in \Family$ the \separation $\sepM{X}{Y}$, and let $X', Y'$ be
the pair realizing the minimum of this quantity.  Merge the two sets
into a new set $Z = X' \cup Y'$, create a new node for this set having
the node for $X'$ and $Y'$ as children, and set its label to be
$\sepM{X'}{Y'}$. Finally, remove $X'$ and $Y'$ from $\Family$ and
insert $Z$ into it. Repeat till there is a single element in
$\Family$. Clearly, the result is a tree that tracks the connected
components of the \MST.

To make the presentation consistent, let $\sepMA{X}{Y}$ be the minimum
$\numA$ such that $\PartAprx{X \cup Y}{1}{\numA}$ is connected.
Computing $\sepMA{X}{Y}$ can be done by computing
$\sepMA{\func}{\funcA}$ for each pair of functions separately. This in
turn, can be done by first computing $\alpha = \sepM{\func}{\funcA}$
and observing that $r$ is between $\alpha/4$ and $\alpha$. In
particular, $r$ must be a power of two, so there are only $3$
candidate values to consider, and which is the right one can be
decided using \lemref{connectivity}.

So, in the above, we use $\sepMA{\cdot}{\cdot}$ instead of
$\sepM{\cdot}{\cdot}$, and let $\HTree$ be the resulting tree.  For a
value $\num$, let $\LHST{\num}$ be the set of nodes such that their
label is smaller than $\num$, but their parent label is larger than
$\num$.  It is easy to verify that $\LHST{\num}$ corresponds to
$\PartAprxC = \PartAprx{\FuncSet}{1}{\num}$; indeed, every cluster $C
\in \PartAprxC$ corresponds to a node $u \in \LHST{\num}$, such that
the set of functions stored in the leaves of the subtree of $u$,
denoted by $\funcsX{u}$ is $\Cluster$. The following can be easily
proved by induction.

\begin{lemma}
    Consider a recursive call \SearchNaive{}$(\FuncSetA,
    \PartitionB, \query)$ made during the search algorithm
    execution. Then $\FuncSetA = \funcsX{u}$, and $\PartitionB =
    \brc{\funcsX{v} \sep{ v \in \LHST{\num} \text{ and } v \text{ is
             in the subtree of } u }}$.
    
    That is, a recursive call of \SearchNaive corresponds to a subtree
    of $\HTree$.
\end{lemma}

Of course, not all possible subtrees are candidates to be such a
recursive call. In particular, \SearchNaive{} can now be interpreted
as working on a subtree $\Tree$ of \HTree, as follows:
\begin{compactenum}[\qquad(A)]
    \item If $\Tree$ is a single node $u$, then find the closet
    function to $\funcsX{u}$. Using the sketch this can be done
    quickly.
    
    \item Otherwise, computes a distance $\numA$, such that the number
    of nodes in the level $\LHSTX{\Tree}{\numA}$ is roughly half the
    number of leaves of $\Tree$.
    
    \item Using interval data-structure determine if the \separation
    $\sepM{\query}{\funcsX{\Tree}}$ is in the range $[\numA/8N, \numA
    8N^2 ]$. If so, we found the desired \ANN.
    
    \item If $\sepM{\query}{\funcsX{\Tree}} > \numA 8 N^2 $ then
    continue recursively on portion of $\Tree$ above
    $\LHSTX{\Tree}{\numA}$.
    
    \item If $\sepM{\query}{\funcsX{\Tree}} < \numA / 8 N $ then we
    know the node $u \in \LHSTX{\Tree}{\numA}$ such that the \ANN
    belongs to $\funcsX{u}$. Continue the search recursively on the
    subtree of $\Tree$ rooted at $u$.
\end{compactenum}
\medskip

That is, \SearchNaive breaks $\Tree$ into subtrees, and continues the
search recursively on one of the subtrees. Significantly, every such
subtree has constant fraction of the size of $\Tree$, and every edge
of $\Tree$ belongs to a single such subtree.

The preprocessing now works by precomputing all the data-structures
required by \SearchNaive. Of course, the most natural approach would
be to precompute \HTree, and build the search tree by simulating the
above recursion on \HTree. Fortunately, this is not necessary, we
simulate running \SearchNaive, and investigate all the different
recursive calls. We thus only use the above \HTree in analyzing the
preprocessing running time. See \figrefpage{search:naive}.

In particular, given a subtree $\Tree$ with $m$ edges, the
corresponding partition $\PartitionB$ would have at most $m$
sets. Each such set would have a $\delta$-sketch, and we compute a
$\eps/8$-sketch for each one of these sketches. Namely, the input size
here is $M = O\pth{ m /\delta^{\constSk}}$. Computing the
$\eps/8$-sketches for each one of these sketches reduces the total
number of functions to $M' = O\pth{ m/\eps^{\constSk}}$, and takes
$U_1 = O\pth{ M / \eps^{\constSk}} = O\pth{ m
   \pth{\eps\delta}^{-\constSk}}$ time, see
\secref{sketch:time}. Computing the splitting distance, using
\lemref{splitting}, takes $U_2 = O\pth{ M' \log M' +
   1/\eps^{\constSk}} = O\pth{ m \eps^{-\constSk} \log m}$ time.
Computing the interval data-structure \lemref{interval} takes $U_3 =
O\pth{M' \eps^{-d-1} \log n \log M'}$ time, and requires $S_1 =
O\pth{M' \eps^{-d-1} \log n }$ space. This breaks $\Tree$ into edge
disjoint subtrees $\Tree_1, \ldots, \Tree_t$, and we compute the
search data-structure for each one of them separately (each one of
these subtrees is smaller by a constant fraction of the original
tree). Finally, we need to compute the $\delta$-sketches for the
clusters sent to the appropriate recursive calls, and this takes $U_4
= O \pth{ M/\delta^{\constSk}}$, by \secref{sketch:time}.

Every edge of the tree $\Tree$ gets charged for the amount of work
spent in building the top level data-structure. That is, the top level
amortized work each edge of $\Tree$ has to pay is
\begin{align*}
    &O\pth{\MakeBig \pth{U_1 + U_2 + U_3 + U_4} / m }%
    \\%
    &= O \pth{ \pth{\eps\delta}^{-\constSk} 
       + \eps^{-\constSk} \log m 
       + \eps^{-d-1-\constSk} \log^2 n 
       + \delta^{-2\constSk} }%
    \\%
    &=O \pth{ \eps^{-2\constSk} \log^{2 \constSk} n},
\end{align*}
assuming $\constSk \geq 2$. Since an edge of $\Tree$ gets charged at
most $O\pth{\log n}$ times by this recursive construction, we conclude
that the total preprocessing time is $O\pth{ n \eps^{-2\constSk}
   \log^{2\constSk+1} n}$.

As for the space, we have by the same argumentation, that each edge
requires $O\pth{ \log n \cdot (S_1/m)} = \pth{ \eps^{-d-1 -\constSk}
   \log^2 n }$. As such, the overall space used by the data-structure
is $\pth{ n\eps^{-d-1 -\constSk} \log^2 n }$. As for the query time,
it boils down to $O(\log n)$ interval queries, and then scanning one
$O(\eps)$-sketch. As such, this takes $O\pth{ \log^2 n +
   1/\eps^{\constSk}}$ time.

\subsection{The result}

\RestateTheorem{main}{\bodyMain}

\begin{proof}
    The query time stated above is $O\pth{ \log^2 n +
       1/\eps^{\constSk}}$. To get the improved query time, we observe
    that \SearchNaive, performs a sequence of point-location queries
    in a sequence of interval near neighbor data-structures (i.e.,
    compressed quadtrees), and then it scans a set of functions of
    size $O\pth{ 1/\eps^{\constSk}}$ to find the \ANN. We take all
    these quadtrees spread through our data-structure, and assign
    them priority, where a quadtree $\QTree_1$ has higher priority
    than a compressed quadtree $\Qtree_2$ if $\QTree_1$ is queried
    after $\QTree_2$ for any search query. This defines an acyclic
    ordering on these compressed quadtrees. Overlaying all these
    compressed quadtrees together, one needs to return for the query
    point, the leaf of the highest priority quadtree that contains
    the query point. This can be easily done by scanning the
    compressed quadtree, and for every leaf computing the highest
    priority leaf that contains it (observe, that here we are
    overlaying only the nodes in the compressed quadtrees that are
    marked by some sublevel set -- nodes that are empty are ignored).
    
    A tedious but straightforward induction implies that doing a
    point-location query in the resulting quadtree is equivalent to
    running the search procedure as described above. Once we found the
    leaf that contains the query point, we scan the sketch associated
    with this cell, and return the computed nearest-neighbor.
\end{proof}

\RestateGeneric{\correfpage{main}}{\bodyMainCor}

\begin{proof}
    We build the data-structure of \thmref{main}, except that instead
    of linearly scanning the sketch during the query time, we
    preprocess each such sketch for an exact point-location query;
    that is, we compute the lower envelope of the sketch and
    preprocess it for vertical ray shooting \cite{ae-rsir-98}.  This
    would require $O\pth{1/\eps^{O(1)}}$ space for each such sketch,
    and the linear scanning that takes $O(1/\eps^{O(1)})$ time, now is
    replaced by a point-location query that takes $O\pth{ \log
       1/\eps^{O(1)}} = O \pth{\log 1/\eps} = O( \log n)$, as desired.
    
    As for the second part, observe that every leaf of the compressed
    quadtree is the set difference of two canonical grid cells. The
    lower envelope of the functions associated with such a leaf,
    induce a partition of this leaf into regions with total complexity
    $O\pth{1/\eps^{O(1)}}$.
\end{proof}

\section{Applications}
\seclab{applications}

We present some concrete classes of functions that satisfy our
framework, and for which we construct \AVD{}'s efficiently.

\subsection{Multiplicative distance functions with %
   additive offsets}
\seclab{mult:plus:offset}

As a warm-up we present the simpler case of additively offset
multiplicative distance functions. The results of this section are
almost subsumed by more general results in \secref{fat:bodies}. Here
the sublevel sets look like expanding balls but there is a time lag
before the balls even come into existence i.e. sublevel sets are empty
up-to a certain level, this corresponds to the additive offsets. In
\secref{fat:bodies} the sublevel sets are more general fat bodies but
there is no additive offset.  The results in the present section
essentially give an \AVD construction of approximate weighted Voronoi
diagrams.  More formally, we are given a set of points $\PntSet =
\brc{\pnt_1,\dots,\pnt_n}$. For $i=1,\ldots, n$, the point $\pnt_i$
has weight $\wt_i >  0$, and a constant $\oset_i \geq 0$ associated
with it. We define $\func_i(\query) = \wt_i \norm{\query-\pnt_i} +
\oset_i$. Let $\FuncSet = \brc{\func_1,\dots,\func_n}$. We have,
$\subLevel{\pth[]{\func_i}}{y} = \emptyset$ for $y < \oset_i$ and
$\subLevel{\pth[]{\func_i}}{y} =
\ball{\pnt_i}{\frac{y-\oset_i}{\wt_i}}$ for $y \geq \oset_i$.
Checking conditions \ccondref{s:computable} and
\ccondref{grid:computable} is trivial. As for
\ccondref{f:s:computable} we have the following easy lemma,

\begin{lemma}%
    \lemlab{intersect:c:o:n:d:2}%
    For any $1 \leq i, j \leq n$ we have
    \begin{align*}
        \sepM{\func_i}{\func_j} = \max \pth{\max(\oset_i,\oset_j),
           \norm{\pnt_i-\pnt_j}\frac{\wt_i \wt_j}{\wt_i + \wt_j} +
           \frac{\oset_i \wt_j + \oset_j \wt_i}{\wt_i + \wt_j}}.
    \end{align*}
\end{lemma}

\begin{proof}
    The $i$\th distance function is $\func_i(\query) = \wt_i
    \norm{\query-\pnt_i} + \oset_i$. As such, for $y <
    \max(\oset_i,\oset_j)$ either $\subLevel{\pth[]{\func_i}}{y} =
    \emptyset$, or $\subLevel{\pth[]{\func_j}}{y} = \emptyset$ and
    $\subLevel{\pth[]{\func_i}}{y} \cap \subLevel{\pth[]{\func_j}}{y}
    = \emptyset$. For $y \geq \max(\oset_i,\oset_j)$, we have
    \begin{align*}
        \func_i(\query) \leq y%
        \quad\implies\quad%
        \wt_i \norm{\query-\pnt_i} + \oset_i \leq y
        \quad\implies\quad%
        \norm{\query-\pnt_i} \leq \frac{y - \oset_i}{\wt_i},
        %
        %
    \end{align*}
    which implies that $\query \in \ball{\pnt_i}{\frac{y -
          \oset_i}{\wt_i}}$; that is, we have
    $\subLevel{\pth[]{\func_i}}{y} = \ball{\pnt_i}{\frac{y -
          \oset_i}{\wt_i}}$ and $\subLevel{\pth[]{\func_j}}{y} =
    \ball{\pnt_j}{\frac{y - \oset_j}{\wt_j}}$.
    
    Now, if $\pnt_i = \pnt_j$ then the \separation distance between the
    two functions is the minimal value such that their sublevel sets
    are not empty, and this is $\max(\oset_i,\oset_j)$. In particular,
    the given expression
    \begin{align*}
        \alpha%
        =%
        \max \pth{\max(\oset_i,\oset_j),
           \norm{\pnt_i-\pnt_j}\frac{\wt_i \wt_j}{\wt_i + \wt_j} +
           \frac{\oset_i \wt_j + \oset_j \wt_i}{\wt_i + \wt_j}}
    \end{align*}
    evaluates to $\max(\oset_i,\oset_j)$, as desired.

    If $\pnt_i \neq \pnt_j$ the sublevel sets intersect for the first
    time when the balls $\ball{\pnt_i}{\frac{y - \oset_i}{\wt_i}}$ and
    $\ball{\pnt_j}{\frac{y - \oset_j}{\wt_j}}$ touch at a point that
    belongs to the segment $\pnt_i \pnt_j$. Clearly then we have
    \begin{align*}
        &\norm{\pnt_i - \pnt_j} = \frac{y - \oset_i}{\wt_i} + \frac{y
           - \oset_j}{\wt_j}%
        \quad\implies\quad
        \wt_i \wt_j \norm{\pnt_i - \pnt_j} = \wt_j\pth{y - \oset_i} +
        \wt_i\pth{y - \oset_j}\\%
        & \quad\implies\quad
        \pth[]{\wt_j + \wt_j} y = \wt_i \wt_j \norm{\pnt_i - \pnt_j} +
        \wt_j\oset_i + \wt_i\oset_j
        \\&%
        \quad\implies\quad
        y = \norm{\pnt_i - \pnt_j}\frac{\wt_i\wt_j}{\wt_i + \wt_j} +
        \frac{\oset_i \wt_j + \oset_j \wt_i}{\wt_i + \wt_j}
    \end{align*}
    \aftermathA
\end{proof}

\newcommand{\BodyLemmaContainCondB}{%
   Given $1 \leq i,j \leq n$ such that $\wt_i \leq \wt_j$. Suppose $y
   \geq \max(\oset_i,\oset_j) $. Then, $\subLevel{\pth[]{\func_j}}{y}
   \subseteq \subLevel{\pth[]{\func_i}}{(1+\delta)y}$ if and only if
   $y \geq \frac{\norm{\pnt_i - \pnt_j} + \oset_i / \wt_i - \oset_j /
      \wt_j} {\pth[]{1+\delta}/{\wt_i} - 1 / \wt_j}$.%
}

\begin{lemma}%
    \lemlab{containCond2}%
    \BodyLemmaContainCondB
\end{lemma}

\begin{proof}
    For $y \geq \max(\oset_i,\oset_j)$ we have that
    $\subLevel{\pth[]{\func_i}}{y} = \ball{\pnt_i}{\frac{y -
          \oset_i}{\wt_i}}$ and $\subLevel{\pth[]{\func_j}}{y} =
    \ball{\pnt_j}{\frac{y - \oset_j}{\wt_j}}$.  If $\pnt_i = \pnt_j$
    then for any $y$ such that $\frac{(1+\delta)y - \oset_i}{\wt_i}
    \geq \frac{y - \oset_j}{\wt_j}$ we will have that
    $\subLevel{\pth[]{\func_j}}{y} \subseteq
    \subLevel{\pth[]{\func_i}}{(1+\delta)y}$.  Clearly this condition
    is also necessary. It is easy to verify that this is equivalent to
    the desired expression.
    
    Consider the case $\pnt_i \neq \pnt_j$.  Sufficiency: Notice that
    for any $\pntA \in \ball{\pnt_j}{\frac{y - \oset_j}{\wt_j}}$ we
    have $\norm{\pntA - \pnt_i} \leq \norm{\pnt_i - \pnt_j} +
    \norm{\pnt_j - \pntA} \leq \norm{\pnt_i - \pnt_j} + \frac{y -
       \oset_j}{\wt_j}$ by the triangle inequality. Therefore, if
    $\frac{(1+\delta)y - \oset_i}{\wt_i} \geq \norm{\pnt_i - \pnt_j} +
    \frac{y - \oset_j}{\wt_j}$, then $\ball{\pnt_j}{\frac{y -
          \oset_j}{\wt_j}} \subseteq \ball{\pnt_i}{ \frac{(1+\delta)y
          - \oset_i}{\wt_i}}$. This is exactly the stated
    condition. Indeed, by rearrangement,
    \begin{align*}
        y\pth{(1+\delta)/\wt_i - 1/\wt_j} \geq \norm{\pnt_i - \pnt_j}
        + \oset_i/\wt_i - \oset_j/\wt_j.
    \end{align*}
    Necessity: Notice that $\ball{\pnt_j}{\frac{y - \oset_j}{\wt_j}}$
    has a boundary point at distance $\frac{y - \oset_j}{\wt_j}$ from
    $\pnt_j$ on the directed line from $\pnt_i$ to $\pnt_j$ on the
    other side of $\pnt_j$ as $\pnt_i$, while
    $\ball{\pnt_i}{\frac{(1+\delta)y - \oset_i}{ \wt_i}}$ has the
    intercept of $\frac{(1+\delta)y - \oset_i}{\wt_i} - \norm{\pnt_i -
       \pnt_j}$.  For the condition to hold it must be true that
    $\frac{(1+\delta)y - \oset_i}{\wt_i} - \norm{\pnt_i - \pnt_j} \geq
    \frac{y - \oset_j}{\wt_j}$, which is also the stated condition.
\end{proof}

It is easy to see that compactness \pcondref{sublevel:compact} and
bounded growth \pcondref{bounded:growth} hold for the set of functions
$\FuncSet$ (for \pcondref{bounded:growth} we can take the growth
function $\grfunc_{\pth[]{\func_i}}(y) = (y - \oset_i)/\wt_i$ for $y
\geq \oset_i$ and the growth constant $\grconst$ to be $2$).  The
following lemma proves the sketch property \pcondref{sketch:small}.
\begin{lemma}%
    \lemlab{conFVerify2}%
    For any $\FuncSetA \subseteq \FuncSet$ and $\delta > 0$ there is a
    $(\delta,y_0)$-sketch $\FuncSetB \subseteq \FuncSetA$ with
    $\cardin{\FuncSetB} = 1$ and $y_0 = 3 \CR{\FuncSetA}
    \cardin{\FuncSetA}/\delta$.
\end{lemma}

\begin{proof}
    If $\cardin{\FuncSetA} = 1$ we can let $\FuncSetB = \FuncSetA$ and
    the result is easily seen to be true. Otherwise, let $l =
    \CR{\FuncSetA}$ for brevity. Observe that $l \geq \max \limits_{i
       : \func_i \in \FuncSetA} \oset_i$, as otherwise some
    $\subLevel{\pth[]{\func_i}}{l} = \emptyset$ and cannot be part of
    a connected collection of sets.  Let $\cardin{\FuncSetA} = m \geq 2$
    and let $\FuncSetA = \brc{\func_1,\dots,\func_m}$, and assume that
    we have $\wt_1 \leq \wt_i, 1 \leq i \leq m$.  We let $\FuncSetB =
    \brc{\func_1}$, the function with the minimum associated weight.
    We are restricted to the range $l \geq \oset_i, 1 \leq i \leq m$
    so, $\subLevel{\pth[]{\func_i}}{l}$ is the ball
    $\ball{\pnt_i}{\frac{l - \oset_i}{\wt_i}}$ for each $1 \leq i \leq
    m$. Since $\subLevel{\FuncSetA}{l}$ is connected, it must be true
    that for $2 \leq j \leq m$ there exist a sequence of distinct
    indices, $1 = i_1 , i_2, \dots, i_{k-1}, i_k = j$ such that
    $\ball{\pnt_{i_r}}{\frac{l - \oset_{i_r}}{ \wt_{i_r}}} \cap
    \ball{\pnt_{i_{r+1}}}{\frac{l - \oset_{i_{r+1}}}{\wt_{i_{r+1}}}}
    \neq \emptyset$ for $1 \leq r \leq k - 1$. By
    \lemref{intersect:c:o:n:d:2} we can write that,
    \begin{align*}
        l \geq \frac{\norm{\pnt_{i_r} - \pnt_{i_{r+1}}} +
           \frac{\oset_{i_r}}{\wt_{i_r}} +
           \frac{\oset_{i_{r+1}}}{\wt_{i_{r+1}}}} {\frac{1}{\wt_{i_r}}
           + \frac{1}{\wt_{i_{r+1}}}}.
    \end{align*}
    Rearranging,
    \begin{align*}
        \norm{\pnt_{i_r} - \pnt_{i_{r+1}}} &\leq
        l\pth{\frac{1}{\wt_{i_r}} + \frac{1}{\wt_{i_{r+1}}}} -
        \pth{\frac{\oset_{i_r}}{\wt_{i_r}} +
           \frac{\oset_{i_{r+1}}}{\wt_{i_{r+1}}}} \\
        &\leq \frac{2l}{\wt_1},
    \end{align*}
    as $\wt_1 \leq \wt_i$, for $1 \leq i \leq m$.  It follows by the
    triangle inequality and the above, that $ \norm{\pnt_{i_1} -
       \pnt_{i_m}} \leq \sum_{r = 1}^{m-1} \norm{\pnt_{i_r} -
       \pnt_{i_{r+1}}} \leq \frac{2(m-1)l}{\wt_1} \leq
    \frac{2ml}{\wt_1}.  $ Thus we have,
    \begin{align}
        \eqlab{llb}%
        \norm{\pnt_1 - \pnt_j} \leq \frac{2ml}{\wt_1},
    \end{align}
    for $j = 1,\dots,m$.  Let $y_0 = \frac{3l
       \cardin{\FuncSetA}}{\delta} = \frac{3 m l}{\delta}$.  Then, for
    $y \geq y_0$ we have that,
    \begin{align*}
        y \geq \frac{3ml}{\delta} = \frac{\frac{2ml}{\wt_1} +
           \frac{ml}{\wt_1}}{\frac{\delta}{\wt_1}}\geq
        \frac{\frac{2ml}{\wt_1} +
           \frac{l}{\wt_1}}{\frac{\delta}{\wt_1}},
    \end{align*}
    for $m \geq 2$. Using \Eqref{llb} and the above, we have for $y
    \geq y_0$ since $l \geq \oset_1$,
    \begin{align*}
        y \geq \frac{\norm{\pnt_1 - \pnt_j} +
           \frac{l}{\wt_1}}{\frac{\delta}{\wt_1}} \geq
        \frac{\norm{\pnt_1 - \pnt_j} +
           \frac{\oset_1}{\wt_1}}{\frac{\delta}{\wt_1}}.
    \end{align*}
    It follows that for $y \geq y_0$,
    \begin{align*}
        y &\geq \frac{\norm{\pnt_1 - \pnt_j} +
           \frac{\oset_1}{\wt_1}}{\frac{\delta}{\wt_1}}%
        \geq \frac{\norm{\pnt_1 - \pnt_j} + \frac{\oset_1}{\wt_1} -
           \frac{\oset_j}{\wt_j}} {\frac{\delta}{\wt_1} +
           (\frac{1}{\wt_1} - \frac{1}{\wt_j})} \\ %
        &=%
        \frac{\norm{\pnt_1 - \pnt_j} + \frac{\oset_1}{\wt_1} -
           \frac{\oset_j}{\wt_j}} {\frac{(1+\delta)}{\wt_1} -
           \frac{1}{\wt_j}},
    \end{align*}
    as $\wt_1 \leq \wt_j$ for $1 \leq j \leq m$. Thus, by
    \lemref{containCond2}, $ \ball{\pnt_{j}}{\frac{y -
          \oset_j}{\wt_j}} \subseteq \ball{\pnt_1}{\frac{(1+\delta)y -
          \oset_1}{\wt_1}}, $ for $y \geq y_0$ and therefore by
    definition, $\FuncSetB$ is a $(\delta,y_0)$-sketch for
    $\FuncSetA$.
\end{proof}


We thus get the following result.
\RestateTheorem{mult:offset:main}{\bodyMultOffestMain}

\subsection{Scaling distance -- generalized %
   polytope distances}%
\seclab{fat:bodies}%

Let $\obj \subseteq \Re^d$ be a compact set homeomorphic to
$\ball{0}{1}$ and containing a ``center'' point $\cntr$ in its
interior. Then $\obj$ is \emphi{star shaped} if for any point $\pntB
\in \obj$ the entire segment $\cntr\pntB$ is also in $\obj$.
Naturally, any convex body $\obj$ with any center $\cntr \in \obj$ is
star shaped.  The \emphi{$t$-scaling} of $\obj$ with a center $\cntr$
is the set $t \obj = \brc{t \pth[]{\pntB - \cntr} + \cntr \sep{ \pntB
      \in \obj}}$.

Given a star shaped object $\obj$ with a center $\cntr$, the
\emphi{scaling distance} of a point $\query$ from $\obj$ is the
minimum $t$, such that $\pnt \in t \obj$, and let
$\objfunc{\obj}(\query)$ denote this distance function.  Note that,
for any $y \geq 0$, the sublevel set
$\subLevel{\pth[]{\objfunc{\obj}}}{y}$ is the $y$-scaling of $\obj$,
that is $\subLevel{\pth[]{\objfunc{\obj}}}{y} = y\obj$.

Note, that for a point $\pnt \in \Re^d$, if we take $\obj =
\ball{\pnt}{1}$ with center $\pnt$, then $\objfunc{\obj}\pth{\query} =
\dist{\pnt}{\query}$. That is, this distance notion is a strict
extension of the Euclidean distance.

Henceforward, for this section, we assume that an object $\obj$
contains the origin in its interior and the origin is the designated
center, unless otherwise stated.

\begin{defn}
    Let $\obj \subseteq \Re^d$ be a star shaped object centered at
    $\cntr$.  We say that $\obj$ is $\fatness$-fat if there is a 
    number $\iradius$ such that,
    $\ballE{\cntr}{\iradius}%
     \subseteq%
     \obj%
     \subseteq%
     \ballE{\cntr}{\fatness \iradius}$.
\end{defn}

\begin{figure}[t]%
    \centerline{%
       \includegraphics[width=0.4\linewidth,clip=]{\si{afat}}%
    }%
    \caption{Being $\fatness$-rounded fat.}%
    \figlab{afat}%
\end{figure}%

\begin{defn}
    \deflab{fat:rounded}
    Let $\obj$ be a star shaped object centered at $\cntr$. 
    We say that $\obj$ is $\fatness$-rounded fat if there is
    a radius $\iradius$ such that,
    \begin{inparaenum}[(i)] 
      \item 
             $\ballE{\cntr}{\iradius} 
              \subseteq
              \obj 
              \subseteq
              \ballE{\cntr}{\fatness \iradius}$ and,
      \item For every point $\pnt$ in the boundary
            of $\obj$, the cone 
            $\CHX{\ballE{\cntr}{\iradius} \cup \pnt}$, lies
            within $\obj$,
    \end{inparaenum}
    see \figref{afat}.
\end{defn}

\begin{figure*}[t]
    \centerline{
       \includegraphics[width=0.4\linewidth,clip=]{\si{fconvex}}
    }
    \caption{A $\fatness$-fat convex body is $\fatness$-rounded fat.}
    \figlab{fconvex}
\end{figure*}

By definition any $\alpha$-rounded fat object is 
also $\alpha$-fat. However, it is not true that
a $\alpha$-fat object is necessarily $\alpha'$-rounded
fat for any $\alpha'$, that is even allowed to depend
on $\alpha$. The following useful result is easy
to see, also see \figref{fconvex} for an illustration.
\begin{lemma}
    \lemlab{fat:conv:cfat} %
    Let $\obj$ be a $\fatness$-fat object. If $\obj$ is convex then
    $\obj$ is also $\fatness$-rounded fat.
\end{lemma}

Given a set $\ObjSet =
\brc{\obj_1,\obj_2,\dots,\obj_n}$ of $n$ star shaped objects,
consider the set $\FuncSet$ of $n$ scaling distance functions, where
the $i$\th function, for $i=1,\ldots, n$ is $\func_i =
\objfunc{\obj_i}$.  We assume that the boundary of each object
$\obj_i$ has constant complexity.

We next argue that $\FuncSet$ complies with the framework of
\secref{conditions}. Using standard techniques, we can compute the
quantities required in conditions
\ccondref{s:computable}--\ccondrefpage{f:s:computable} including the
diameter of the sublevel set $\diameter{y\obj_i} = y
\diameter{\obj_i}$.  Also, trivially we have that condition
\pcondrefpage{sublevel:compact} is satisfied as the sublevel sets are
dilations of the $\obj_i$ and are thus compact by definition.  The
next few lemmas establish that both bounded growth
\pcondref{bounded:growth} and the sketch property
\pcondref{sketch:small} are also true, if the objects 
are also $\fatness$-rounded fat for some constant
$\fatness$.
\begin{lemma}%
    \lemlab{fatexpand}%
    Given $\fatness > 0$, suppose $\obj$ is
    a star shaped object that is $\fatness$-rounded fat.  Then for
    any $\constC \geq 2\fatness$ and any $y
    \geq 0, \eps > 0$ we have that $y\obj \Mplus
    \ball{0}{(\nfrac{\eps}{\constC}) \diameter{y\obj}} \subseteq
    (1+\eps)y\obj$; that is, $\subLevel{\pth[]{\objfunc{\obj}}}{y}
    \Mplus \ball{0}{(\nfrac{\eps}{\constC})\diameter{\subLevel{
             \pth[]{\objfunc{\obj}}}{y}}} \subseteq
    \subLevel{\pth[]{\objfunc{\obj}}}{(1+\eps)y}$.
\end{lemma}

\begin{figure*}[t]
    \centerline{
       \includegraphics[width=0.7\linewidth,clip=]{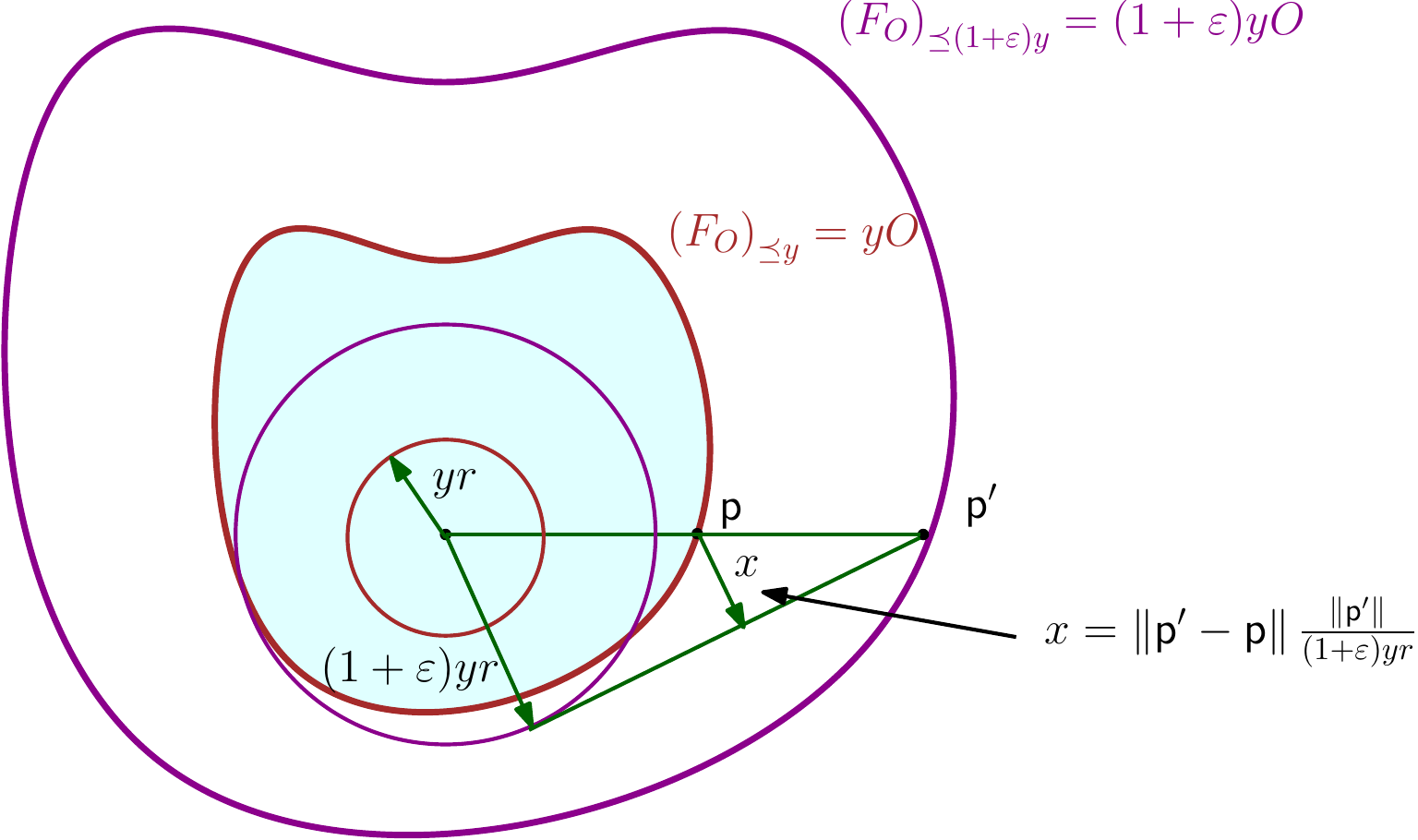}
    }
    \caption{The $(1+\eps)$ expansion of $y \obj$ contains
             $\ball{\pnt}{x}$.}
    \figlab{fig:contain}
\end{figure*}

\begin{proof}
    Since $\subLevel{\pth[]{\objfunc{\obj}}}{y} = y \obj$ we show that
    $y \obj \Mplus \ball{0}{(\nfrac{\eps}{\constC})\diameter{y\obj}}
    \subseteq (1+\eps)y \obj$. Let $\iradius$ be the
    radius guaranteed by \defref{fat:rounded} for $\obj$. 
    Clearly $\diameter{y\obj} = y \diameter{\obj} \leq 2 y \fatness \iradius$.
    We show that for every $\pnt \in \bdry{y \obj}$
    we have that $\pnt + \ball{0}{(\nfrac{\eps}{c})\diameter{y\obj}}
    \subseteq (1+\eps)y\obj$ \footnote{ Topological arguments can show
       that for objects homeomorphic to balls, if this is true for
       boundary points $\pnt$, then for any $\pnt \in \obj$, $\pnt +
       \ball{0}{(\nfrac{\eps}{c})\diameter{y\obj}} \subseteq
       (1+\eps)y\obj$. We omit the technical argument here.  }.  Let
    $\pnt \in \bdry{y \obj}$. It is sufficient to show that,
    $\ball{\pnt}{(\nfrac{2\eps y \fatness \iradius}{\constC})} \subseteq
    (1+\eps)y\obj$.  Clearly $\pnt' = (1+\eps)\pnt \in
    \bdry{(1+\eps)y\obj}$. Since the cone,
    $\CHX{\ball{\cntr}{(1+\eps)y\iradius} \cup \pnt'}$ is
    in $(1+\eps)y\obj$, it is clear that the ball of radius,
    \[
      x
      = 
      \norm{\pnt' - \pnt} \frac{\norm{\pnt'}}{(1+\eps)y\iradius}
      =
      \norm{\pnt' - \pnt} \frac{\norm{\pnt}}{y\iradius}
    \]
    is completely within $(1+\eps)y\obj$, see
    \figref{fig:contain}.  Now, $\norm{\pnt - \pnt'} = \eps
    \norm{\pnt} \geq \eps y \iradius$, and $\norm{\pnt} \geq y \iradius$.
    It follows that $x \geq \eps y \iradius$. 
    If we choose $\constC \geq 2\fatness$, the claimed result is easily
    seen to hold.
\end{proof}

\begin{figure}[t]
    \centerline{
       \includegraphics{\si{b_example}}
    }%
    \caption{The object $\obj$ is $\fatness$-fat but not $\fatness'$-rounded
             fat.  In particular, the point $\pnt$ is in $\obj \Mplus
             \ball{0}{(\nfrac{\eps}{c}){\diameter{\obj}}}$ but not in
             $(1+\eps)\obj$. In particular, the scaling distance function is
             discontinuous at $\pntA$.}
    \figlab{fig:bexample}
\end{figure}

By the above lemma we can take the growth function
$\grfunc_{\objfunc{\obj_i}}(y)%
=%
\diameter{\subLevel{\pth[]{\objfunc{\obj_i}}}{y}} =%
y \diameter{\obj_i}$ and the growth constant, see
\pcondrefpage{bounded:growth}, for the set of functions
$\objfunc{\obj_i}$ to be $\grconst = \constC = 2 \fatness$. 
If the object $\obj$ is $\fatness$-fat but not
$\fatness'$-rounded fat for any constant $\fatness' > 0$ then
it may be that its scaling distance function grows arbitrarily quickly and
thus fails to comply with our framework, see \figref{fig:bexample}.
It is not hard to see that \lemref{fatexpand} implies that bounded
growth \pcondref{bounded:growth} is satisfied for all the functions
$\func_1, \ldots, \func_n$ when the objects under consideration
$\obj_1, \ldots, \obj_n$ are $\fatness$-rounded fat.  To show that
condition \pcondref{sketch:small} is satisfied, is slightly harder.

\begin{lemma}%
    \lemlab{ucover}%
    Let $\ObjSet$ be a set of $n$ star shaped objects $\obj_1, \obj_2,
    \dots, \obj_n$.  Let $\fatness \geq 1$ 
    be any constant. Suppose that $\obj_1, \ldots, \obj_n$ are
    $\fatness$-rounded fat. Then, for any $\delta > 0$, there is a
    subset $\IndSet \subseteq \brc{1,2,\dots,n}$ with
    $\cardin{\IndSet} = O(\delta^{-d})$, such that for all $y \geq 0$,
    we have
    \begin{align*}
        \bigcup_{i \in [n]} y \obj_i \subseteq \bigcup_{i \in \IndSet}
        (1+\delta)y\obj_i.
    \end{align*}
    Moreover, for every $i \in \IndSet$ we have that
    $\diameter{\obj_i} = \Omega ( \max_i $ $\diameter{\obj_i} )$.
\end{lemma}

\begin{proof}
    Recall our convention, that here all the bodies are centered at the
    origin. Clearly it is sufficient to show this for $y = 1$. Let
    $\iradius_i$ for $i = 1,\dots, n$ be the radius of the ball
    satisfying the conditions of \defref{fat:rounded} for 
    object $\obj_i$.
    Assume that $\iradius_1 \geq
    \iradius_2 \dots \geq \iradius_n$. If 
    $\fatness \iradius_j \leq \iradius_1$
    for some $j$, then $\obj_j, \dots, \obj_n$ are contained in
    $\obj_1$, we can add $1$ to the set $\IndSet$. From now we assume
    that for each $1 \leq j \leq n$, $\fatness \iradius_j > \iradius_1$ 
    i.e. we can ignore the sufficiently small objects. 
    Our index set $\IndSet$ is a
    subset of these prefix indices for which, 
    $\fatness \iradius_i \geq \iradius_1$. As such,
    \begin{align*}
        \diameter{\obj_i} \geq 2 \iradius_i%
                          \geq \frac{2}{\fatness} \iradius_1%
                          \geq \frac{4}{\fatness^2} (2 \fatness \iradius_1)%
                          \geq \frac{4}{\fatness^2} \max_{1 \leq i \leq n} \diameter{\obj_i},
    \end{align*}
    for any $i \in \IndSet$.
    It is easy to see that, $ \bigcup_{i \in [n]} \obj_i \subseteq
    \bigcup_{i \in [n]} \ball{0}{\fatness \iradius_i} \subseteq
    \ball{0}{\fatness \iradius_1}.  $ We tile the ball 
    $\ball{0}{\fatness \iradius_1}$
    with cubes of diameter at most $\delta \fatness \iradius_1 /\constC'$ where
    $\constC'$ is a constant that we determine shortly.  Notice that
    the number of such cubes is $O(\delta^{-d})$. Let $\CellSet$
    denote the set of these cubes.  If $\cell \cap \obj_i \neq
    \emptyset$ for some $1 \leq i \leq n$ and $\cell \in \CellSet$
    then we add $\cell$ to a set $\CellSetA$ and $i$ to our index set
    $\IndSet$. Notice that we choose at most one object among all objects
    that might intersect $\cell$. Now, $ \bigcup_{i \in [n]} \obj_i
    \subseteq \bigcup_{\cell \in \CellSetA}\cell$, as 
    $\bigcup_{\cell \in \CellSetA} \cell$ covers 
    $\ball{0}{\fatness \iradius_1}$. Observe,
    $\cardin{\IndSet} \leq \cardin{\CellSetA} \leq \cardin{\CellSet} =
    \order{\delta^{-d}}$.  We show that it is possible to choose
    $\constC'$ so large that, $ \bigcup_{\cell \in \CellSetA} \cell
    \subseteq \bigcup_{i \in \IndSet} (1+\delta) \obj_i.  $ Since $\cell
    \cap \obj_i \neq \emptyset$ and $\diameter{\cell} \leq
    \delta\fatness \iradius_1/\constC'$, $\cell \subseteq \obj_i \Mplus
    \ball{0}{\nfrac{\delta\fatness \iradius_1}{\constC'}}$. We choose
    $\constC'$ large enough so that 
    $\delta \fatness \iradius_1 / \constC' \leq
    \delta \diameter{\obj_i}/\constC$ where 
    $\constC = 2 \fatness$ is the constant from
    \lemref{fatexpand}. Then we will have by \lemref{fatexpand},
    \begin{align*}
        \cell%
        &\subseteq%
        \obj_i \Mplus \ball{0}{\delta \fatness \iradius_1 / \constC'}%
        \subseteq%
        \obj_i \Mplus \ball{0}{\delta \diameter{\obj_i} / \constC}%
        \\%
        &\subseteq%
        (1+\delta) \obj_i,
    \end{align*}
    proving the claim. Now,
    \begin{align*}
        \delta \fatness \iradius_1 / \constC'%
        \leq%
        \delta \fatness^2 \iradius_i / \constC' \leq%
        \frac{\delta \fatness^2 \diameter{\obj_i}}{2 \constC'}%
        \leq%
        \frac{\delta \diameter{\obj_i}}{\constC},
    \end{align*}
    if $\constC' = \nfrac{\constC \fatness^2}{2} =
    \fatness^3$.
\end{proof}

\begin{lemma}%
    \lemlab{smalldisp}%
    Let $\fatness \geq 1$  be any constant. Let
    $\obj$ be a star shaped object that is $\fatness$-rounded fat,
    and let $\delta > 0$.  Let $\pntA \in \Re^d$ with $\norm{\pntA}
    \leq \delta \diameter{\obj} / \constC$ where $\constC = 2
       \fatness$. Then we have that $\obj + \pntA
    \subseteq (1+\delta) \obj$.
\end{lemma}

\begin{proof}
    We have, $\obj + \pntA \subseteq \obj \Mplus
    \ball{0}{\norm{\pntA}} \ \subseteq \obj \Mplus \ball{0}{\delta
       \diameter{\obj} / \constC}$ as $\norm{\pntA} \leq \delta
    \diameter{\obj} / \constC$. The result follows by appealing to
    \lemref{fatexpand}.
\end{proof}

\begin{lemma}%
    \lemlab{connlb}%
    For $i=1,\ldots, n$, let $\obj_i$ be a star shaped object in
    $\Re^d$ centered at a point $\pnt_i$. Let $\ObjSet = \brc{\obj_1,
       \dots, \obj_n}$, $\PntSet = \brc{\pnt_1, \dots, \pnt_n}$, and
    $\FuncSet = \brc{\func_i \sep 1 \leq i \leq n}$, where $\func_i =
    \objfunc{\obj_i}$, for $i=1,\ldots, n$.  For $i=1,\ldots, n$, let
    $\iradius_i$ denote the radius of the ball for $\obj_i$ from
    \defref{fat:rounded}, and let $\iradius = \max_i \iradius_i$. Then, 
    $\CR{\FuncSet} \geq \nfrac{\diameter{\PntSet}}{(2 n \fatness \iradius)}$.
\end{lemma}

\begin{proof}
    The claim is trivially true if $\diameter{\PntSet} = 0$, i.e. all
    the points $\pnt_i$ are the same.  Let $l = \CR{\FuncSet}$, for
    brevity.  As we have $\subLevel{\pth[]{\func_i}}{l} = 
    l \obj_i$, where the scaling is done around its center $\pnt_i$, 
    it follows that the sets $l \obj_i$ for
    $i = 1, \dots, n$ are connected.  Since $l \obj_i 
    \subseteq \ball{\pnt_i}{l \fatness \iradius_i} \subseteq \ball{\pnt_i}{l
       \fatness \iradius}$ it is easy to see that the balls 
    $\ball{\pnt_i}{l \fatness \iradius}$ for $i = 1, \dots, n$ are also
    connected. Let $\pntA, \pntB \in
    \PntSet$ be such that $\norm{\pntA - \pntB} = \diameter{\PntSet}$.
    There is a sequence of distinct $i_1, \dots, i_k \in \brc{1,\dots,
       n}$ such that $\pntA = \pnt_{i_1}, \pntB = \pnt_{i_k}$ and we
    have $\ball{\pnt_{i_r}}{l \fatness \iradius} \cap 
          \ball{\pnt_{i_{r+1}}}{l \fatness \iradius} 
          \neq \emptyset$ for $1 \leq r \leq k - 1$. It follows
    that $\norm{\pnt_{i_r} - \pnt_{i_{r+1}}} \leq 2 l \fatness \iradius, 1 \leq
    r \leq k - 1$. By the triangle inequality,
    \begin{align*}
        \diameter{\PntSet}%
        &=%
        \norm{\pntA - \pntB}%
        =%
        \norm{\pnt_{i_1} - \pnt_{i_k}}%
        \leq%
        \sum_{r = 1}^{k-1} \norm{\pnt_{i_r} - \pnt_{i_{r+1}}}%
        \\%
        &\leq%
        \sum_{r=1}^{k-1} 2 l \fatness \iradius%
        =%
        2(k-1)l \fatness \iradius%
        \leq%
        2 n l \fatness \iradius,
    \end{align*}
    thus proving the claim.
\end{proof}

We can now show that condition \pcondrefpage{sketch:small} holds for
the $\objfunc{\obj_i}$.

\begin{lemma}
    \lemlab{fat:sk:small}%
    Consider the setting of \lemref{connlb}.  Given $\delta > 0$,
    there is a index set $\IndSet \subseteq \brc{1,\dots, n}$ with
    $\cardin{\IndSet} = \order{\delta^{-d}}$ and $y_0 = O(l \cdot
    n/\delta)$ such that the functions $\brc{\func_j \sep{ j \in
          \IndSet }}$ form a $(\delta,y_0)$-sketch, where $l =
    \CR{\FuncSet}$.
\end{lemma}

\begin{proof}
    We provide a sketch of the proof as details are easy but tedious.
    For each $1 \leq i,j \leq n$ we consider the set of objects
    $\obj_{ij} = \obj_i + \pnt_j -\pnt_i$, i.e.  $\obj_{ij}$ is
    $\obj_i$ translated so that it is centered at $\pnt_j$.  By
    \lemref{ucover} there is an index set $\IndSet \subseteq
    \brc{1,\dots,n}$ with $\cardin{\IndSet} = \order{\delta^{-d}}$ such that
    for all $y$ and any fixed $j$ with $1 \leq j \leq n$ we have that
    $ \bigcup_{i \in [n]} y \obj_{ij} \subseteq \bigcup_{i \in
       \IndSet}(1+\delta/4) y \obj_{ij}$. Let $\iradius_i$ denote the 
    radius of the ball for $\obj_i$ from \defref{fat:rounded}, 
    and let $\iradius = \max_i \iradius_i$. By
    \lemref{connlb}, we have that,
    $l \geq \nfrac{\diameter{\PntSet}}{(2 n \fatness \iradius)}$. 
    \lemref{ucover} finds a $\IndSet$ such that for all $i \in
    \IndSet$, $\iradius_i \geq \ordergeq{\iradius}$.  A translated copy
    $\obj_{ij} = \obj_i + \pnt_j - \pnt_i$ is a translation by a
    vector $\pntA = \pnt_j - \pnt_i$.  As $l \geq
    \nfrac{\diameter{\PntSet}}{(2 n \fatness \iradius)}$, 
    there is a $y_0 =
    \order{l n / \delta}$ such that $\norm{\pnt_j - \pnt_i} \leq
    \delta \diameter{y_0 \obj_i} / 4 \constC$ for all $1 \leq i,j \leq
    n$, where $\constC = 2 \fatness$. Thus
    using \lemref{smalldisp}, $(1+\delta/4) y_0 \obj_i + (\pnt_j -
    \pnt_i) \subseteq (1+\delta/4)^2 y_0 \obj_i \subseteq (1+\delta)
    y_0 \obj_i$. Clearly this also holds for any $y \geq y_0$.  Thus
    for $y \geq y_0$ we have $(1+\delta)y \obj_i, $ covers $y \obj_i +
    (\pnt_j - \pnt_i)$ for $1 \leq i \leq n$.  It is then easy to see
    that $\brc{\func_i \sep i \in \IndSet}$ is a $(\delta,y_0)$-sketch.
\end{proof}

We conclude that for $\fatness$-rounded fat objects, 
the scaling distance function they define falls 
under our framework. We thus get the following result.


\RestateGeneric{\thmrefpage{a:fat}}{\bodyAFat}

Note, that the result in \thmref{a:fat} covers any symmetric convex
metric. Indeed, given a convex symmetric shape $C$ centered at the
origin, the distance it induces for any pair of points $\pnt, \pntA
\in \Re^d$, is the scaling distance of $C$ centered $\pnt$ to $\pntA$
(or, by symmetry, the scaling distance of $\pnt$ from $C$ centered at
$\pntA$). Under this distance $\Re^d$ is a metric space, and of
course, the triangle inequality holds. By an appropriate scaling of
space, which does not affect the norm (except for scaling it) we can
make $C$ fat, and now \thmref{a:fat} applies.  Of course,
\thmref{a:fat} is considerably more general, allowing each of the
points to induce a different scaling distance function, and the
distance induced does not have to comply with the triangle inequality.

\subsection{Nearest furthest-neighbor}
\seclab{f:n:neighbor}

For a set of points $\PntSetC \subseteq \Re^d$ and a point $\query$,
the \emphi{furthest-neighbor distance} of $\query$ from $\PntSetB$, is
$\fn_{\PntSetC}(\query) = \max_{\pntC \in \PntSetC}
\dist{\query}{\pntC}$; that is, it is the furthest one might have to
travel from $\query$ to arrive to a point of $\PntSetC$.  For example,
$\PntSetC$ might be the set of locations of facilities, where it is
known that one of them is always open, and one is interested in the
worst case distance a client has to travel to reach an open facility.
The function $\fn_{\PntSetC}(\cdot)$ is known as the
\emphi{furthest-neighbor Voronoi} diagram, and while its worst case
combinatorial complexity is similar to the regular Voronoi diagram, it
can be approximated using a constant size representation (in low
dimensions), see \cite{h-caspm-99}.

Given $n$ sets of points $\PntSet_1, \dots, \PntSet_n$ in $\Re^d$, we
are interested in the distance function $\fn(\query) = \min_i
\fn_i(\query)$, where $\fn_i(\query) = \fn_{\PntSet_i}\pth{\query}$. 
This quantity arises natural when one
tries to model uncertainty; indeed, let $\PntSet_i$ be the set of
possible locations of the $i$\th point (i.e., the location of the
$i$\th point is chosen randomly, somehow, from the set
$\PntSet_i$). Thus, $\fn_i(\query)$ is the worst case
distance to the $i$\th point, and $\fn(\query)$ is the worst-case
nearest neighbor distance to the random point-set generated by picking
the $i$\th point from $\PntSet_i$, for $i=1,\ldots, n$. We refer to
$\fn(\cdot)$ as the \emphi{nearest furthest-neighbor} distance, and we
are interested in its approximation.

\remove{
Given a query point $\query$ we want to compute approximately the radius
$\radA$ of the smallest ball centered at $\query$ such that
$\ball{\query}{\radA}$ contains at least one of the sets completely,
that is there is some $i$ with $1 \leq i \leq n$ such that $\PntSet_i
\subseteq \ball{\query}{\radA}$. More formally, we have that
\[
\radA = \min_{i = 1, \ldots ,n} \max_{\pntA \in \PntSet_i}
\dist{\query}{\pntA}.
\]
}

A naive solution to this problem would maintain a data structure for
computing the furthest neighbor approximately for each of the
$\PntSet_i$ and then just compute the minimum of those distances. A
data-structure to compute a $1 - \eps$ approximation to the furthest
neighbor takes $O(1/\eps^d)$ space for $O(1/\eps^d)$ query time, see
\cite{h-caspm-99} although this was probably known before. Thus the
entire data structure would take up total space of $O(n / \eps^d)$
with a query time of $O(n / \eps^d)$. By using our general framework
we can speed up the computation.
We will show that $\fn_i$, for $i=1,\ldots, n$
satisfy the conditions \pcondref{sublevel:compact} --
\pcondref{sketch:small} and
\ccondref{s:computable}--\ccondref{f:s:computable}. By \thmref{main}
we can prepare a data-structure of size $O(n \polylog(n))$ that allows
us a query time of $O(\log n)$ to find the desired nearest furthest-neighbor
approximately. In
order to facilitate the computations of \separation{}s we also
maintain data structures for $(1 - \eps/4)$-approximate furthest
neighbor search for each of the point sets $\PntSet_i$ for $i = 1, 2,
\dots, n$ where $\eps$ is the approximation parameter for
approximating the nearest furthest-neighbor,
i.e. the approximation parameter for
the problem we are trying to solve. Also, for $\mu = \eps^2 / 144$ we
also maintain $\mu$-coresets for computing the minimum enclosing ball
(\emphi{\MEB{}}) approximately for each $\PntSet_i$ for $i = 1, 2,
\dots, n$.  Each such coreset has $O(1/\eps^2)$ points, see
\cite{bhi-accs-02,bc-scsb-03, bc-ocsb-03}.  For each $i$ with $1 \leq
i \leq n$, the radius of the \MEB of the coreset points is a $(1 +
\mu)$-approximation to the radius of the \MEB of $\PntSet_i$.

\subsubsection{Satisfaction of the conditions}
\begin{observation}
    \obslab{fnn:sl:struc} We have that
    $\subLevel{{\pth[]{\fn_i}}}{y} = \bigcap \limits_{\pntA \in
       \PntSet_i} \ball{\pntA}{y}$, and
    $\diameter{\subLevel{{\pth[]{\fn_i}}}{y}} \leq 2y$.
\end{observation}
Given the above observation, it is easy to see that Condition
\pcondref{sublevel:compact} is true, as
$\subLevel{\pth[]{\fn_i}}{y}$ is a finite intersection of compact
sets. The following Lemma shows that Condition
\pcondref{bounded:growth} is also true, by letting the growth function
$\grfunc_{\pth[]{\fn_i}}(y) = y$. Since $y \geq
\diameter{\subLevel{{\pth[]{\fn_i}}}{y}}/2$ by
\obsref{fnn:sl:struc}, it follows that we can choose the growth
constant $\grconst$ to be $2$.
\begin{lemma}
    For any $i$ with $1 \leq i \leq n$, if
    $\subLevel{\pth[]{\fn_i}}{y} \neq \emptyset$, it is true that,
    \[
    \subLevel{\pth[]{\fn_i}}{y}%
    \Mplus%
    \ball{0}{\eps y}%
    \subseteq%
    \subLevel{\pth[]{\fn_i}}{(1+\eps)y}.
    \]
\end{lemma}

\begin{proof}
    Consider any point $\query$ in $\subLevel{\pth[]{\fn_i}}{y}%
    \Mplus%
    \ball{0}{\eps y}%
    $.  It is easy to see that $\ball{\query}{(1+\eps)y} \supseteq
    \PntSet_i$ by the triangle inequality, and so $\query \in
    \subLevel{\pth[]{\fn_i}}{(1+\eps)y}$.
\end{proof}

Condition \pcondref{sketch:small} is implied by the following,

\begin{lemma}
    Let $\FuncSetA \subseteq \brc{\fn_1,\dots,\fn_n}$ denote any set of
    functions. Then, given any $\delta > 0$, there is a subset
    $\FuncSetB \subseteq \FuncSetA$ with $\cardin{\FuncSetB} = 1$ and
    a $y_0$ with $y_0 = O(\CR{\FuncSetA} \cardin{\FuncSetA} / \delta)$
    such that $\FuncSetB$ is a $(\delta,y_0)$-sketch for $\FuncSetA$.
\end{lemma}

\begin{proof}
    Without loss of generality, let $\FuncSetA = \brc{\fn_1, \dots,
       \fn_m}$ where $m = \cardin{\FuncSetA}$. Let $z =
    \CR{\FuncSetA}$. Since $\subLevel{\pth[]{\fn_i}}{z}$ for $i =
    1,\dots, m$ are all connected, and by
    \obsref{fnn:sl:struc}, for each $i$ with $1 \leq i \leq m$ we have
    that $\diameter{\subLevel{\pth[]{\fn_i}}{z}} \leq 2z$ it follows
    that for any two points $\pntA \in \PntSet_j, \pntB \in \PntSet_k$
    with $1 \leq j, k \leq m$ there are points $\pntA' \in
    \subLevel{\pth[]{\fn_j}}{z}, \pntB' \in
    \subLevel{\pth[]{\fn_k}}{z}$, such that $\dist{\pntA}{\pntA'}
    \leq z$, $\dist{\pntB}{\pntB'} \leq z$ (by definition of the
    function $\fn_j$ and $\fn_k$ respectively) and
    $\dist{\pntA'}{\pntB'} \leq 2 m z$, by the bound on the diameter
    of the sublevel sets and the condition of being connected, which 
    is the same as the intersection graph of the sets being connected.  
    It follows by the triangle inequality that
    $\dist{\pntA}{\pntB} \leq 2(m+1)z \leq 4mz$ i.e.
    $\diameter{\PntSet_1 \cup \dots \cup \PntSet_m} \leq 4mz$.  Let
    $\FuncSetB$ be a set containing an arbitrary function from $\FuncSetA$ say,
    $\FuncSetB = \brc{\fn_1}$.  It is not too hard to see (or one
    can apply \lemref{containCond2} with the case $\alpha_i = 0, \wt_i
    = 1$), that for every $i = 1,\dots,m$, $\ball{\pntB}{y} \subseteq
    \ball{\pntA}{(1+\delta)y}$ for any points $\pntA, \pntB \in
    \PntSet_1 \cup \dots \cup \PntSet_m$ if $y \geq y_0 = 4mz /
    \delta$, from the bound on the diameter of $\PntSet_1 \cup \dots
    \cup \PntSet_m$. Thus, for any $i$ with $1 \leq i \leq m$ we have
    that,
    \[
    \subLevel{\pth[]{\fn_i}}{y} = \bigcap_{\pntB \in \PntSet_i}
    \ball{\pntB}{y} \subseteq%
    \ball{\pntA}{(1+\delta)y},
    \]
    for all $\pntA \in \PntSet_1$ and $y \geq y_0$. As such,
    \[
    \subLevel{\pth[]{\fn_i}}{y} \subseteq \bigcap_{\pntA \in
       \PntSet_1} \ball{\pntA}{(1+\delta)y} =
    \subLevel{\pth[]{\fn_1}}{(1+\delta)y},
    \]
    and the result follows.
\end{proof}

\begin{remark}
    For $\FuncSetA = \brc{\fn_1, \dots, \fn_m}$, notice that we
    can compute the above set $\FuncSetB$ in $O(1)$ time and that we
    can compute a polynomial approximation to $\CR{\FuncSetA}$ in
    $O(m)$ time, since we can compute the diameter
    $\diameter{\PntSet_1 \cup \dots \cup \PntSet_m}$ approximately in
    $O(m)$ time - we simply take an arbitrary point in $\PntSet_1$ and
    compute furthest distances approximately for each of $\PntSet_i$
    for $1 \leq i \leq m$ and take the maximum of these. We can use
    the $O(1)$ time query algorithm for furthest neighbor for this
    purpose.
\end{remark}
We now consider the computability conditions \ccondref{s:computable}--
\ccondref{f:s:computable}. To compute $\sepM{\query}{\fn_i}$ we use
the data structure for approximate furthest neighbor queries to get a
$(1-\eps/4)$-approximation to this number. We run the preprocessing
algorithm, see \secref{build}, with approximation parameter
$\eps/4$. By \remref{grid:set:bound}, we only
tile the sublevel sets $\subLevel{\pth[]{\fn_i}}{y}$ with canonical
cubes of size (rounded to a power of two) $\eps
\grfunc_{\pth[]{\fn_i}}(y) / 4 = \eps y/4$.  Notice that the minimum
$y$ such that $\subLevel{\pth[]{\fn_i}}{y}$ is non-empty is clearly
the radius of the \MEB of the point set $\PntSet_i$ and for this value
$\subLevel{\pth[]{\fn_i}}{y}$ just includes the center of the
\MEB. Let $\mebc_i$ and $\mebr_i$ denote the center and radius of the
exact \MEB, and $\mebc'_i$ and $\mebr'_i$ denote those computed by
using the coreset.  Since $\mebr'_i \leq (1 + \mu) \mebr_i$, it not
too hard to see that $\norm{\mebc_i - \mebc'_i} \leq 3 \sqrt{\mu}
\mebr_i = \eps \mebr_i/4$.  This is implied for example by
\lemref{meb:fact}, presented in \apndref{meb:fact:proof}, which may
also be of independent interest (this assumes $\mu < 1$ which is
indeed true).  We are required to tile the sublevel set
$\subLevel{\pth[]{\fn_i}}{y}$ for some $y \geq \mebr'_i/(1 + \mu)$
using cubes of size roughly $\eps y/4$, but we use in fact cubes of
size roughly $\eps y / \constC$ for some large constant $\constC$.
One can consider such cubes at increasing distance from the point
$\mebc'_i$. Choosing any point within a cube one evaluates approximately
the furthest neighbor distance of $\PntSet_i$, and checks if it is at
most $y(1+O(\eps))$.  If so, one includes the cube. Since all such
cubes will intersect the ball around $\mebc'_i$ of radius $y$, or a
slight expansion of it, the number of such cubes is still
$O(1/\eps^d)$. Now, the procedure in fact guarantees that all subcubes
intersecting $\subLevel{\pth[]{\fn_i}}{y}$ are found, but in fact
there may be some that do not intersect it. However, this is not a
problem as such cells will still be inside
$\subLevel{\pth[]{\fn_i}}{(1+\eps/4)y}$ which is what is really
required.  To see that this works should be intuitively clear. We omit
the straightforward, but tedious proof.  This settles Condition
\ccondref{grid:computable}.  Notice that the \separation between
$\fn_i$ and $\fn_j$ is the radius of the minimum enclosing ball of
the point set $\PntSet_i \cup \PntSet_j$. Using the $(1 +
\mu)$-coresets that we have for the \MEB of $\PntSet_i$ and
$\PntSet_j$ we can compute a $(1 + 2\mu)$-coreset for the \MEB of
$\PntSet_i \cup \PntSet_j$ by simply merging those coresets. This
allows us to approximately compute the \separation.
\begin{remark}
    For the computability conditions
    \ccondref{s:computable}--\ccondref{f:s:computable} we only showed
    approximate results, that is the \separation{}s were computed
    approximately. In fact, to be conservative, we used $\eps/4$ as
    the approximation parameter in the construction algorithm and the
    furthest neighbor data structure. As a tedious but straightforward
    argument can show, the main lemmas \lemref{near:neighbor} and
    \lemref{interval} for near neighbor and interval range queries as
    well as the ones for computing the connectivity and splitting
    radius \lemref{connectivity} and \lemref{splitting} can work under
    such approximate computations, with the same running times.
\end{remark}
We thus get the following result.

\RestateGeneric{\thmrefpage{fnn:main}}{\bodyFNNMain}

\begin{proof}
    We only need to show how get the improved space and query
    time. Observe that every one of the sets $\PntSet_i$ can be
    replaced by a subset $\PntSetC_i \subseteq \PntSet_i$, of size
    $O(1/\eps^d \log (1/\eps))$, such that for any point $\query \in
    \Re^d$, we have that $\fn_{\PntSetC_{i}}(\query) \leq
    \fn_{\PntSet_{i}}(\query) \leq
    (1+\eps/4)\fn_{\PntSetC_{i}}(\query)$. Such a subset can be
    computed in $O(\cardin{\PntSet_i} )$ time, see
    \cite{h-caspm-99}\footnote{One computes an appropriate exponential
       grid, of size $O(1/\eps^d \log (1/\eps) )$, and pick from each
       grid cell one representative point from the points stored
       inside this cell.}. %
    We thus perform this transformation for each one of the uncertain
    point sets $\PntSet_1, \ldots, \PntSet_n$, which reduces the input
    size to $O(n/\eps^d \log(1/\eps))$. We now apply our main result
    to the distance functions induced by the reduced sets
    $\PntSetC_1,\ldots, \PntSetC_n$.
\end{proof}

\section{Conclusions}
\seclab{conclusions}

In this paper, we investigated what classes of functions have
minimization diagrams that can be approximated efficiently -- where our
emphasis was on distance functions. We defined a general framework
and the requirements on the distance functions to fall under it.  For
this framework, we presented a new data-structure, with near linear
space and preprocessing time. This data-structure can evaluate
(approximately) the minimization diagram of a query point in
logarithmic time.  Surprisingly, one gets an \AVD (approximate Voronoi
diagram) of this complexity; that is, a decomposition of space with
near linear complexity, such that for every region of this
decomposition a single function serves as an \ANN for all points in
this region.

We also showed some interesting classes of functions for which we get
this \AVD. For example, additive and multiplicative weighted distance
functions. No previous results of this kind were known, and even in
the plane, multiplicative Voronoi diagrams have quadratic complexity
in the worst case (for which the \AVD generated has near linear
complexity for any constant dimension).  The framework also works for
Minkowski metrics of fat convex bodies, and nearest
furthest-neighbor. However, our main result applies to even more
general distance functions.

Several questions remain open for further research:
\begin{compactenum}[(A)]
    \item Are the additional polylog factors in the space necessary?
    In particular, it seems unlikely that using \WSPD{}'s directly, as
    done by Arya and Malamatos \cite{am-lsavd-02}, should work in the
    most general settings, so reducing the logarithmic dependency
    seems quite interesting.  Specifically, can the Arya and Malamatos
    construction \cite{am-lsavd-02} be somehow adapted to this
    framework, possibly with some additional constraints on the
    functions, to get a linear space construction?
    
    \item On the applications side, are constant degree polynomials a
    good family amenable to our framework?  Specifically, consider a
    polynomial $\tau(x)$ that is positive for all $x \geq 0$. Given a
    point $\pntA$, we associate the distance function $f( \query) =
    \tau\pth{\dist{\query}{\pntA}}$ with $\pntA$.  Given a set of such
    distance functions, under which conditions, can one build an \AVD
    for these functions efficiently? (It is not hard to see that in
    the general case this is not possible, at least under our
    framework.)
    
\end{compactenum}

\bibliographystyle{alpha}%
\bibliography{wann}%

\appendix

\section{Bounding the size of intersection of balls %
   of the same radius}
\apndlab{meb:fact:proof}

\begin{lemma}
    \lemlab{meb:fact}%
    Let $\mebc, \mebr$ be the center and radius of the \MEB for a set
    of points $\PntSet = \brc{\pnt_1,\dots,\pnt_m} \subseteq \Re^d$.
    Let $\delta \geq 0$ be any number. Let $\pnt \in \bigcap_{i=1}^m
    \ball{\pnt_i}{(1+\delta)\mebr}$.  Then,
    \[
    \delta \mebr%
    \leq%
    \norm{\pnt - \mebc}%
    \leq%
    \sqrt{4\delta + 2\delta^2} \mebr.
    \]
\end{lemma}
\begin{proof}
    The first inequality follows from the triangle inequality.  We use
    the fact that there are affinely independent points from $\PntSet$, 
    on the
    surface of the \MEB, such that $\mebc$ lies in their convex
    hull. Thus, assume without loss of generality that there are
    points $\pnt_1, \dots, \pnt_k \in \PntSet$ that are affinely
    independent with $\norm{\pnt_i - \mebc} = \mebr$ and $\lambda_i
    \geq 0$ for $i = 1,2, \dots, k$ such that,
    \[
    \mebc = \sum_{i=1}^k \lambda_i \pnt_i, \quad \quad \sum_{i=1}^k
    \lambda_i = 1.
    \]
    We restrict our attention only to the points $\pnt_1, \dots,
    \pnt_k$ since the region $\bigcap_{i=1}^k
    \ball{\pnt_i}{(1+\delta)\mebr}$ contains the region
    $\bigcap_{i=1}^m \ball{\pnt_i}{(1+\delta)\mebr}$.  Consider an
    arbitrary point $\pnt \in \bigcap_{i=1}^k
    \ball{\pnt_i}{(1+\delta)\mebr}$. Let $\pnt'$ be the projection of
    $\pnt$ to the affine subspace spanned by $\pnt_1,\dots,
    \pnt_k$. We first bound $\norm{\pnt' - \mebc}$.  It is easy to see
    that $\pnt' - \mebc$ satisfies $\dotP{\pnt' - \mebc}{\pnt_i} \leq
    0$ for some $i$ with $1 \leq i \leq k$. Without loss of generality
    assume $i = 1$. It follows that,
    \[
    \norm{\pnt' - \pnt_1}%
    \geq%
    \sqrt{\norm{\pnt' - \mebc}^2 + \norm{\mebc - \pnt_1}^2}%
    \geq%
    \sqrt{\mebr^2 + \norm{\pnt' - \mebc}^2}.
    \]
    On the other hand it must be the case that, 
    $\norm{\pnt' - \pnt_1} \leq \norm{\pnt - \pnt_1} \leq (1+\delta)\mebr$. 
    As such, $(1+\delta)\mebr \geq
    \sqrt{\mebr^2 + \norm{\pnt' -\mebc}^2}$, and we have that
    $\norm{\pnt' - \mebc} \leq \sqrt{2\delta + \delta^2}\mebr$.  We
    also have that,
    \[
    (1 + \delta)^2\mebr^2%
    \geq%
    \norm{\pnt - \pnt_1}^2%
    =%
    \norm{\pnt - \pnt'}^2 + \norm{\pnt' - \pnt_1}^2%
    \geq%
    \norm{\pnt - \pnt'}^2 + \mebr^2,
    \]
    implying that $\norm{\pnt - \pnt'} \leq \sqrt{2\delta + \delta^2}
    \mebr$.  It follows by the Pythagorean theorem,
    \[
    \norm{\pnt - \mebc}^2%
    =%
    \norm{\pnt - \pnt'}^2 + \norm{\pnt' - \mebc}^2%
    \leq%
    2(2\delta + \delta^2)\mebr^2,
    \]
    and thus $\norm{\pnt - \mebc} \leq \sqrt{4\delta +
       2\delta^2}\mebr$.
\end{proof}

\section{Basic properties of the functions}
\apndlab{p:proof}

\begin{lemma}%
    \lemlab{s:l:zero:trivial}%
    Let $\FuncSet$ be a set of functions that satisfy the compactness
    \pcondref{sublevel:compact} and bounded growth
    \pcondref{bounded:growth} conditions. Then, for any $\func \in
    \FuncSet$, either $\subLevel{\func}{0} = \emptyset$ or
    $\subLevel{\func}{0}$ consists of a single point.
\end{lemma}

\begin{proof}
    If $\subLevel{\func}{0}$ contains at least two points, then by
    compactness \pcondref{sublevel:compact} of $\subLevel{\func}{0}$
    there are two points $x,y \in \subLevel{\func}{0}$ such that
    $\norm{x - y} = \diameter{\subLevel{\func}{0}} > 0$. By the
    bounded growth \pcondref{bounded:growth} it follows that
    \begin{align*}
        \subLevel{\func}{0} \subseteq \subLevel{\func}{0} \Mplus
        \ball{0}{\frac{\norm{x-y}}{\grconst}} \subseteq
        \subLevel{\func}{0} \Mplus \ball{0}{\grfunc_\func(0)}
        \subseteq \subLevel{\func}{0},
    \end{align*}
    using $\eps = 1$ and the fact that $\grfunc_\func(0) \geq
    \diameter{\subLevel{\func}{0}}/\grconst = \norm{x - y}/\grconst$.
    Thus, $\subLevel{\func}{0} \Mplus
    \ball{0}{\frac{\norm{x-y}}{\grconst}} = \subLevel{\func}{0}$.
    Clearly in $y \Mplus \ball{0}{\frac{\norm{x-y}}{\grconst}}$ there
    is some $y'$ such that $\norm{x - y'} > \norm{x - y}$ which
    contradicts that $x$ and $y$ is a diametrical pair in
    $\subLevel{\func}{0}$.
\end{proof}

By the above lemma, we may assume that a symbolic perturbation
guarantees that $\sepM{\func}{\funcA} > 0$ for $\func \neq
\funcA$. With this convention we have the following,
\begin{observation}%
    \obslab{c:l:zero}%
    If $\CR{\FuncSetA} = 0$ for any non-empty subset $\FuncSetA$ then
    $\cardin{\FuncSetA} = 1$.
\end{observation}

We also assume that the quantities $\sepM{\func}{\funcA}$ are distinct
for all distinct pairs of functions.

\begin{lemma}%
    \lemlab{segment:contained}%
    Let $\func \in \FuncSetA$ and $y \geq 0$. Suppose $\pntA, \pntB
    \in \subLevel{\func}{y}$. Then, $\pntA \pntB \subseteq
    \subLevel{\FuncSetA}{(1+\grconst/2)y}$, where $\pntA \pntB$
    denotes the segment joining $\pntA$ to $\pntB$.
\end{lemma}

\begin{proof}
    If $\pntA = \pntB$, the claim is obvious.  Using bounded growth
    \pcondref{bounded:growth} with $\eps = \grconst/2$, and the
    inequality $\grfunc_\func(y) \geq
    \diameter{\subLevel{\func}{y}}/\grconst$, it follows that
    $\subLevel{\func}{y} \Mplus
    \ball{0}{\diameter{\subLevel{\func}{y}}/2} \subseteq
    \subLevel{\func}{(1+\grconst/2)y}$. Thus, $\pntA \Mplus
    \ball{0}{\diameter{\subLevel{\func}{y}}/2} \subseteq
    \subLevel{\func}{(1+\grconst/2)y}$ as well as
    $\ball{\pntB}{\diameter{\subLevel{\func}{y}}/2} \subseteq
    \subLevel{\func}{(1+\grconst/2)y}$. Since $\norm{\pntA - \pntB}
    \leq \diameter{\subLevel{\func}{y}}$ it follows that the entire
    segment $\pntA \pntB$ is in $\subLevel{\func}{(1+\grconst/2)y}$.
\end{proof}

\begin{lemma}%
    \lemlab{segment:compact:cover}%
    Let $\setA_1, \dots, \setA_m \subseteq \Re^d$ be compact connected 
    sets. Let $\pntA \pntB$ be any segment. Suppose that 
    $\pntA \pntB \cap \setA_i \neq \emptyset$ for all 
    $1 \leq i \leq k$ and $\pntA \pntB
    \subseteq \bigcup_{i=1}^k \setA_i$. Then, the sets $\setA_i, 1
    \leq i \leq k$, are connected.
\end{lemma}

\begin{proof}
    It is sufficient to prove the claim for $\setA_i \subseteq \pntA
    \pntB$, as the truth of the claim for compact sets $\setA_i \cap
    \pntA \pntB$ implies the truth for $\setA_i$.  Thus, assume
    $\setA_i \subseteq \pntA \pntB$.  Suppose the claim is
    false. Consider the intersection graph of the $\setA_i, 1 \leq i
    \leq k$. This graph has at least two components by assumption. Let
    $\setB_1, \dots, \setB_l$ be the partition of $[1,k]$ that define
    these components i.e. for each $1 \leq i \leq l$, the sets
    $\setA_j, j \in \setB_i$ are connected, and $\setA_x \cap \setA_y =
    \emptyset$ for $1 \leq x, y \leq k$ if $x, y$ belong to different
    $\setB_i$. Denote by $\setC_i = \bigcup_{j \in \setB_i} \setA_j$
    for $1 \leq i \leq l$.  Clearly each $C_i$ is compact. By an easy
    compactness argument, there are distinct $1 \leq i_1, i_2 \leq l$
    such that for points $\pntC \in \setC_{i_1}, \pntD \in
    \setC_{i_2}$, we have that $0 < \norm{\pntC - \pntD} =
    \min\limits_{1 \leq x \neq y \leq l, \pnt \in \setC_x, \query \in
       \setC_y} \norm{\pnt - \query}$.  However, this is impossible as
    $\pntC, \pntD$ are distinct points on $\pntA \pntB$ and the
    segment $\pntC \pntD$ is therefore covered by the $\setC_i, 1 \leq
    i \leq l$. It follows that a smaller distance between distinct
    $\setC_i$ must be attainable.
\end{proof}

\begin{lemma}%
    \lemlab{sketch:conn:level}%
    Suppose we are given $\FuncSetB \subseteq \FuncSetA \subseteq
    \FuncSet$, $\delta \geq 0$ and $y \geq 0$, and $\FuncSetB$ is a
    $(\delta,y)$-sketch for $\FuncSetA$. Then, $\CR{\FuncSetB} \leq
    (1+\delta)(1+\grconst/2)\max(y,\CR{\FuncSetA})$.
\end{lemma}
\begin{proof}
    Assume that $\FuncSetA = \brc{\func_1, \dots, \func_m}$ and
    $\FuncSetB = \brc{\func_1, \dots, \func_k}$ where $k \leq m$. If
    $m = 1$ then $k = 1$ and we have by definition $\CR{\FuncSetA} =
    \CR{\FuncSetB} = 0$ and the result clearly holds true. If $m > 1$,
    we need to show that $\subLevel{\pth[]{\func_i}}{y'}$, for
    $i=1,\ldots, k$, are connected, where $y' =
    (1+\delta)(1+\grconst/2)l$ and $l = \max(y,\CR{\FuncSetA})$.  Now
    by definition, $\subLevel{\FuncSetA}{l}$ is a connected set.
    Consider any $1 \leq i \neq j \leq k$. Then there is a
    sequence of distinct indices $i = i_1, i_2, \dots, i_s = j$ such
    that $\subLevel{\pth[]{\func_{i_r}}}{l} \cap
    \subLevel{\pth[]{\func_{i_{r+1}}}}{l} \neq \emptyset$ for $1 \leq
    r \leq s-1 $. Consider any such index say $i_r$ such that $i_r >
    k$ i.e. $\func_{i_r} \notin \FuncSetB$.  Since,
    $\subLevel{\pth[]{\func_{i_r}}}{l} \cap
    \subLevel{\pth[]{\func_{i_{r-1}}}}{l} \neq \emptyset$ and
    $\subLevel{\pth[]{\func_{i_r}}}{l} \cap
    \subLevel{\pth[]{\func_{i_{r+1}}}}{l} \neq \emptyset$ we can
    choose points $\pntA \in \subLevel{\pth[]{\func_{i_{r-1}}}}{l}
    \cap \subLevel{\pth[]{\func_{i_r}}}{l}$ and $\pntB \in
    \subLevel{\pth[]{\func_{i_r}}}{l} \cap
    \subLevel{\pth[]{\func_{i_{r+1}}}}{l}$. Now the entire segment
    $\pntA \pntB \subseteq
    \subLevel{\pth[]{\func_{i_r}}}{(1+\grconst/2)l}$ by
    \lemref{segment:contained}. Since $(1+\grconst/2)l \geq y$ it
    follows by the sketch property \pcondref{sketch:small}, that
    $\pntA \pntB \subseteq
    \subLevel{\pth[]{\func_{i_r}}}{(1+\grconst/2)l} \subseteq
    \subLevel{\FuncSetB}{(1+\grconst/2)(1+\delta)l}$. By
    \lemref{segment:compact:cover} the sets in the 
    minimal cover of $\pntA \pntB$
    by the sublevel sets
    $\subLevel{\pth[]{\func_i}}{(1+\grconst/2)(1+\delta)l}, 1 \leq i
    \leq k$, are connected. It follows that
    $\subLevel{\pth[]{\func_{i_r}}}{(1+\grconst/2)l}$ can be replaced
    by a sub-collection of the
    $\subLevel{\pth[]{\func_{i}}}{(1+\grconst/2)(1+\delta)l}, 1 \leq i
    \leq k$ and the property of neighbor intersections is still valid
    in the chain. We replace each occurrence of the set
    $\subLevel{\pth[]{\func_{i_r}}}{(1+\grconst/2)l}$ for $i_r > k$ by
    the corresponding chain. It is easy to see that the resulting
    chain connects up
    $\subLevel{\pth[]{\func_{i_1}}}{(1+\grconst/2)(1+\delta)l}$ and
    $\subLevel{\pth[]{\func_{i_s}}}{(1+\grconst/2)(1+\delta)l}$. Now,
    duplicate elements can be easily removed without affecting the
    neighbor intersection property of the chain.
\end{proof}

The following testifies that a sketch approximates the \separation of
a set of functions.

\begin{lemma}%
    \lemlab{sketch:rule}%
    Let $\FuncSetB \subseteq \FuncSetA$ be sets of functions, where
    $\FuncSetB$ is a $(\delta,y_0)$-sketch for $\FuncSetA$ for some
    $\delta \geq 0$ and $y_0 \geq 0$.  Let $\query$ be a point such
    that $\sepM{\query}{\FuncSetA} \geq y_0$. Then we have that
    $\sepM{\query}{\FuncSetB} \leq (1 + \delta)
    \sepM{\query}{\FuncSetA}$.
\end{lemma}

\begin{proof}
    Let $l = \sepM{\query}{\FuncSetA}$ and let $\func \in \FuncSetA$
    be a witness that $\query \in \subLevel{\func}{l}$. As $l \geq
    y_0$ we have that $\subLevel{\func}{l} \subseteq \bigcup_{\funcA
       \in \FuncSetB} \subLevel{\funcA}{(1+\delta)l}$ by the sketch
    property (\defref{sketch:def}). As such there is some function
    $\funcA \in \FuncSetB$ such that $\query \in
    \subLevel{\funcA}{(1+\delta)l}$. It follows that
    $\sepM{\query}{\funcA} \leq (1+\delta) \sepM{\query}{\FuncSetA}$.
\end{proof}

\end{document}